\documentclass[a4paper]{article}
\usepackage{authblk} 
\usepackage[pages=all, color=black, position={current page.south}, placement=bottom, scale=1, opacity=1, vshift=5mm]{background}
\SetBgContents{
}      

\usepackage[margin=1in]{geometry} 

\usepackage{amsmath}
\usepackage{amsthm}
\usepackage{amssymb}
\usepackage{lipsum}
\usepackage{xcolor}

\usepackage[utf8]{inputenc}
\usepackage[hidelinks]{hyperref}
\hypersetup{
	unicode,
	pdfauthor={Author One, Author Two, Author Three},
	pdftitle={A simple article template},
	pdfsubject={A simple article template},
	pdfkeywords={article, template, simple},
	pdfproducer={LaTeX},
	pdfcreator={pdflatex}
}

\usepackage{natbib}
\theoremstyle{plain}
\newtheorem{theorem}{Theorem}[section] %
\newtheorem*{proof*}{Proof} %
\newtheorem*{theorem*}{Theorem}
\newtheorem{lemma}[theorem]{Lemma}

\newtheorem{proposition}[theorem]{Proposition} %

\theoremstyle{definition}
\newtheorem{definition}[theorem]{Definition}

\newtheorem{remark}[theorem]{Remark}

\usepackage{graphicx, color}
\graphicspath{{Figures/}}

\usepackage{tikz}
\usetikzlibrary{calc}
\usetikzlibrary{positioning}
\usetikzlibrary{fit}
\usetikzlibrary{shapes,arrows}

\usepackage{algorithm, algorithmic} 
\usepackage{mathrsfs} 
\usepackage{subcaption}
\usepackage{bm}
\usepackage{xcolor}

\usepackage[overlay, absolute]{textpos}
\usepackage{xurl}
\usepackage[toc,page]{appendix}  
\usepackage{titletoc}
\usepackage[acronym]{glossaries}

\newcommand{\1}{\mathbf{1} }
\newcommand{\R}{\mathbb{R}}
\newcommand{\E}{\mathbb{E}}
\newcommand{\strictorder}{\1_{\alpha \leq \theta_1 \leq \cdots\leq \theta_n \leq \beta}}

\renewcommand{\P}{\mathcal{P}}
\newcommand{\betaperiode}{\mathcal{B}e_{(0, \beta-\alpha)}}

\graphicspath{{images/}}

\title{A Projection-Based Approach to Bayesian Age Estimation under Stratigraphic Constraints}

\author[1]{Imène Bouafia}
\author[1]{Anne Philippe}
\author[2]{Guillaume Guérin}

\affil[1]{Laboratoire de Mathématiques Jean Leray, UMR 6629  Nantes Université - CNRS, 44000 Nantes, France}
\affil[2]{Géosciences Rennes, UMR 6118 Univ Rennes - CNRS, 35000 Rennes, France}

\date{}

\begin{document}
\maketitle

\begin{abstract}

	In archaeological chronology building, stratigraphic constraints arising from the excavation context are commonly incorporated into Bayesian models through order-constrained priors.
However, these models suffer from well-documented artifacts: estimated
durations are inflated, and the bias worsens as more samples are dated.
First, we prove that this phenomenon is intrinsic to the prior: when all true ages coincide, the posterior duration converges to the prior support's width. We then propose an
unconstrained posterior that does not include stratigraphic knowledge and project it onto
the space induced by the stratigraphic order. The
projected posterior is the Wasserstein-closest constrained distribution;
it can be computed via isotonic regression on a directed acyclic graph and corrects duration inflation unlike constrained priors. 
Simulation studies and two archaeological case
studies demonstrate substantial bias reduction and precision gains.
   
\end{abstract}

\tableofcontents
\clearpage

\section{Introduction}

Beyond merely contextualizing material remains, numerical dating in
archaeology essentially serves three purposes: 1) evaluate the timing
of behavioral and societal changes at one given site; 2) compare
evidence occurring at distant sites, to study peopling dynamics,
cultural transitions, human settlement patterns, the spread of
cultural practices, etc.; 3) place human remains and behaviors in
their environmental and climatic context, i.e.\ make the link between
archaeological evidence and independently dated proxies, such as, on
the one hand, human evolution and, on the other, ice cores recording past climate fluctuations. Nowadays, prehistoric archaeology is largely driven by the links between human and climate evolution, to estimate to what extent environmental fluctuations influence subsistence strategies in particular, and more generally cultural changes. In this article, we focus on the two latter issues - because they arguably are the most frequently encountered situations in archaeology. 

In practice, during archaeological excavations, samples $(c_1, \ldots, c_n)$ collected at different stratigraphic positions are dated using methods such as, for example, radiocarbon dating ($^{14}$C) (see \cite{hajdas_radiocarbon_2021}), which is based on radioactive decay, and Optically Stimulated Luminescence (OSL) (see \cite{murray_optically_2021}), which is known to be less precise than ($^{14}$C) but can be used to date much older materials.

Because they are based on the measurement of physical quantities, numerical dating methods come with uncertainties that hamper interpretations. To mitigate the issue, research efforts focus on reducing measurement uncertainties (e.g., \cite{fewlass2019pretreatment}, \cite{deviese2018new}), improving the calibration curve for radiocarbon (see, e.g., \cite{heaton2020intcal20}), increasing the number of dated samples, or statistical modelling. In particular, to leverage the spatial information obtained during the excavation, more reliable estimates of the true ages $(\theta_1, \ldots, \theta_n)$  can be obtained using statistical methods that incorporate stratigraphic constraints: the age $\theta_j$ of a given sample must be older than that of a sample lying above it. The principle of stratigraphy leads to an order structure, either total when all samples are in ascending order, or at least partial when batches of samples belong to the same stratigraphic unit inside of which their order is unknown.

The stratigraphic information is routinely incorporated into Bayesian chronological models because it provides valuable prior knowledge. Such models are now standard practice in radiocarbon and luminescence dating and are implemented in archaeological and geochronological applications through software packages such as ChronoModel
(see \cite{lanos_event_2018}), BayLum (see \cite{combes_bayesian_2017}), or OxCal (see \cite{bronk_ramsey_bayesian_2009}). The resulting chronologies directly inform archaeological interpretations through the estimation of phase durations, temporal gaps and rates of cultural change (\cite{philippe2020archaeophases}).

Nevertheless, these approaches have been demonstrated to suffer from systematic statistical artifacts and estimation biases. Early discussions of these issues \cite{steier_use_2000} identified shift phenomena arising when the number of measurements increases or when ages values $\theta \in \R^n$ are close or equal. This is a practical case, called high resolution sampling, when samples are being collected at very close stratigraphic positions on a small section. 

More recently, \cite{guerin_conflict_2026} provided a comprehensive analysis revealing that despite using different software implementations, all major Bayesian chronological model frameworks rely on the same underlying mathematical model, the uniform order prior (see for more details \cite{david_order_2004}), for handling stratigraphic constraints: a prior distribution that assigns zero probability to any chronology violating stratigraphic order. This approach seems natural but has revealed systematic problems. High resolution sampling leads to systematic inflation of estimated durations $d_n = \max\limits_{i} \theta_i - \min\limits_i \theta_i$, because of the uniform order prior. The result is that estimates of phase durations, occupation lengths and temporal gaps are systematically biased. These distortions are not merely technical artifacts, they also affect archaeological interpretations. Consequently, even moderate biases introduced by order constrained priors may propagate into substantial conclusions regarding settlement dynamics, cultural transitions, or rates of social change. A particular undesirable feature is that such biases are increased when the number of measured samples is increased: the more efforts are made to address one research question from a dating perspective, the worse the estimates become. Finally, this problem is also made worse when uncertainties are large.

One might propose abandoning stratigraphic constraints, but they are observations and should therefore not be treated as nuisances. The real challenge is: \textit{how can we enforce stratigraphic ordering without distorting the information carried by the measurements ?}
In other words, the issue is not whether stratigraphic information should be incorporated, but how. Our objective is therefore to enforce stratigraphic consistency while introducing as little additional information as possible beyond that contained in the dating measurements.

Additionally, current constrained Bayesian models require specialized Markov Chain Monte-Carlo (MCMC) schemes to ensure that posterior samples satisfy the imposed order 
throughout the simulation process, for instance \cite{yu2025bayesian}. Although effective, these approaches often involve truncated conditional distributions and carefully tuned proposal mechanisms, which can substantially increase implementation complexity and computational cost. As chronological models grow in size and complexity, these difficulties become increasingly pronounced.

A recent framework by \cite{astfalck2026posteriorprojectioninferenceconstrained} proposes a general strategy for handling constraints in Bayesian Inference: estimate an unconstrained posterior, then project onto the constraint-compatible space. 
We adapt this framework to geochronology, where it addresses a specific and previously under-recognized problem: the duration inflation bias induced by stratigraphic ordering. 
Our approach proceeds in two steps: \begin{enumerate}
    \item Estimate an unconstrained posterior from the measurement model, ignoring stratigraphic constraints.
    \item Project each posterior sample onto the space of chronologies consistent with stratigraphic constraints.
\end{enumerate}
The projected posterior has a principled justification: it is the Wasserstein-closest distribution to the unconstrained posterior among all distributions supported on the constrained chronology space. In other words, we enforce stratigraphic order while making the minimal distortion to the information carried by the measurements. 
The key contribution is showing that this separation, when applied to chronological inference, substantially reduces the duration inflation bias while preserving the archeological information in the measurements. To implement this practically, we encode stratigraphic constraints as a directed acyclic graph and reformulate the projection as an isotonic regression problem, enabling efficient computation via existing block-merging algorithms.

\noindent Our contributions are the following. First, we provide the first rigorous formalization of the duration-inflation
phenomenon reported empirically in  \cite{steier_use_2000}:
Theorem~\ref{thm:posterior_simplex} shows that, under the uniform-order prior and
when all true ages coincide, the posterior duration converges in probability
to the maximal value permitted by the prior support. The bias is therefore
driven by the prior rather than by the data, and it \emph{worsens} as the
number of dated samples increases. Second, we adapt the posterior-projection
framework of \cite{astfalck2026posteriorprojectioninferenceconstrained} to stratigraphic constraints
encoded by a directed acyclic graph: the projection reduces to a weighted
isotonic regression problem and the projected posterior is
characterized as the Wasserstein projection of the unconstrained posterior
onto the constrained space (Proposition~\ref{prop:projected_sampling}).
Theorem~\ref{thm:projected_posterior} then establishes that, in the same worst-case
scenario as Theorem~\ref{thm:posterior_simplex}, the projected posterior displays
the opposite asymptotic behavior: its interior coordinates concentrate.
Finally, we demonstrate the archaeological consequences on two case
studies: a re-analysis of the Çatalhöyük East radiocarbon dataset, and a
luminescence chronology of the Middle Palaeolithic sequence of Gatzarria
Cave, where the gain in precision allows, for the first time on the basis
of numerical dating, an association of the Quina Mousterian with a climatic event.

The remainder of the paper is organized as follows.
Section~\ref{sec:background} formalizes the chronological inference problem,
reviews the order-constrained Bayesian models, and establishes the
duration-inflation result.
Section~\ref{sec:projection} develops the projected-posterior framework for
stratigraphic constraints: we derive the projection onto the admissible
space as a weighted isotonic regression problem, specialize the posterior
projection theory to this setting, illustrate its geometry on a synthetic
two-dimensional example, and describe the computational algorithm.
Section~\ref{sec:simulations} evaluates the proposed
methodology through simulation studies based on the design of
\cite{steier_use_2000}, comparing its statistical and computational
performance with existing approaches. Section~\ref{sec:applications}
presents the two archaeological case studies, and
Section~\ref{sec:conclusion} concludes with a discussion of limitations
and perspectives.

\section{Background} \label{sec:background}
Suppose we have $n$ dated samples yielding age measurements $M_1,
\cdots M_n$, which constitute the observations in the statistical
model.  We assume these observations arise from a parametric family of
distribution $\lbrace f(\cdot \mid \theta),  \; \theta \in \Theta
\subset \R^n \rbrace$. 
The stratigraphic constraints are encoded in a DAG $G=(V,E)$, where the
vertices $V=\{1,\dots,n\}$ represent the samples and each edge
$(i,j)\in E$ encodes the observed stratigraphic relation ``sample $i$ is
older than or contemporaneous with sample $j$'', i.e.\ the constraint
$\theta_i\le\theta_j$. We write $E=\{e_1,\dots,e_m\}$. The associated constrained space  \begin{equation}
    \mathcal{P} = \lbrace \theta \in \R^n \ \mid \theta_i \leq \theta_j \ \forall \  (i, j) \in E \rbrace,
\end{equation}
also called an isotonic cone, which is a closed convex cone. \\
We choose a prior distribution  whose support is  the constrained space
$\mathcal{P}$, its density is of the form 
\begin{equation}
    \pi_\P(\theta) \propto \pi(\theta) \1_{\theta \in \P} \quad \forall
                     \theta \in \R^n
  \end{equation}
  where  $\pi$ is a density  supported on $\R^n$.
  The prior density can be rewritten as 
    \begin{equation}
  \pi_\P(\theta) \propto \pi(\theta) \prod_{i= 1}^{n}  \1_{\lbrack \ell_i, L_i \rbrack}(\theta_i) \label{eq:constrainedprior} \quad 
    \text{ with } 
    \begin{cases}
        \ell_i &=  \max \lbrace \theta_j  \ : \  j \in \mathrm{Pa}(i) \rbrace \\
        L_i &=  \min \lbrace \theta_j \ :  \ j  \in \mathrm{Ch}(i) \rbrace,
    \end{cases}
  \end{equation}
where the sets $\mathrm{Ch}$ and $\mathrm{Pa}$ define the children and parents of $i \in V$: 
\begin{equation*}
    \mathrm{Ch}(i) = \lbrace j \in V \mid (i,j) \in E \rbrace \qquad
    \mathrm{Pa}(i) = \lbrace j \in V \mid (j,i) \in E \rbrace. 
\end{equation*}

\noindent Therefore, the posterior distribution has a density
satisfying, for all   $\theta \in \R^n$ 
\begin{equation} \label{posterior}
    \pi_\P(\theta \mid M) \propto f(M \mid \theta) \pi_\P(\theta). 
\end{equation}

\noindent In the special case of a chain order \[
\P = \lbrace \theta \in \R^n \mid \theta_1 \leq \cdots \leq \theta_n \rbrace, \]
the prior definition \eqref{eq:constrainedprior} becomes:
\begin{align}
\pi_\P(\theta) &= \pi(\theta) \; \1_{\theta_1 \leq \theta_2} \; \1_{\theta_{n-1}\leq \theta_n} \; \prod_{i=2}^{n-1} \1_{\theta_{i-1} \leq \theta_i \leq \theta_{i+1}} \\
&= \pi(\theta) \; \1_{\theta_1\leq \cdots \leq \theta_n}.
\end{align}
A usual framework is to consider a known interval (geochronological period) $[\alpha, \beta] \subset \R$ and then construct a uniform order on this interval, allowing the constrained prior density to be given by:
\begin{equation} \label{eq:uodensity}
\pi_\P(\theta) = \frac{n!}{(\beta-\alpha)^n} \strictorder \quad \forall \theta \in \R^n.
\end{equation}

\noindent The artifacts arise from the resulting distribution of the total duration 
\[ d_n = \theta_n - \theta_1.\]
Under the uniform order prior, this duration follows a Beta distribution:
\begin{equation} \label{eq:duration}
 d_n \sim \betaperiode(n-1, 2), 
\end{equation}
where $\mathcal{B}e_{(\alpha, \beta)}(a,b)$ represents a Beta distribution $X \sim \mathcal{B}e(a,b)$ with the affine transformation $$ U = (\beta-\alpha) X + \alpha,$$
so that the density would be defined over the interval $[\alpha, \beta]$:
\begin{align*}
    f_{U}(u) &= \frac{1}{Be(a,b)}(u-\alpha)^{a-1} (\beta-u)^{b-1} \frac{1_{\lbrace u \in (\alpha, \beta) \rbrace}}{(\beta-\alpha)^{b+a-1} }.
\end{align*}
The expectation and variance reveal the problematic behavior:
\begin{align}
\mathbb{E}[d_n] &= (\beta-\alpha) \frac{n-1}{n+1} \xrightarrow{n \to \infty} \beta-\alpha, \\
\text{Var}(d_n) &= (\beta-\alpha)^2 \frac{2(n-1)}{(n+1)^2 (n+2)} \xrightarrow{n \to \infty} 0.
\end{align}
The distribution of the duration is strongly skewed toward its maximum possible value. Besides, if we draw the duration distribution, we observe that most of the probability mass concentrates near $(\beta-\alpha)$ over the interval $\lbrack 0, (\beta-\alpha) \rbrack$. This bias toward larger durations systematically underestimates the youngest coordinate $\theta_1$ and overestimates the oldest coordinate $\theta_n$.

\noindent Therefore, the range $d_n$ converges in probability to its maximal
possible value, $\beta-\alpha$, as $n\to\infty$. This means that the
uniform prior on $\mathcal P$ asymptotically favors vectors whose
range is close to its maximum.

\noindent The following theorem shows that the posterior distribution inherits
the same asymptotic behavior when the true ages are all equal.
\begin{theorem}\label{thm:posterior_simplex}

Consider Gaussian measurements, i.e. the likelihood function $f(\cdot \mid \theta)$ is the density of the Gaussian vector $\mathcal{N}(\theta, \sigma^2 I_n)$ in \eqref{posterior} with the prior distribution being \eqref{eq:uodensity}. \\
Assume that  the true data-generating model is
\begin{equation}\label{true}
M_i=\theta^0+\sigma \varepsilon_i,\qquad i=1,\ldots,n,
\end{equation}
where \(\theta^0\in(\alpha,\beta)\) and where
$\varepsilon_1,\ldots,\varepsilon_n$ are independent and identically
distributed standard normal random variables. Then, for  every \(\delta\in (0,\beta-\alpha)\), 
$$
\Pi_{\cal P }\!\left( \theta_n-\theta_1>(\beta-\alpha) -\delta \,\middle|\,M \right) \xrightarrow{\mathbb{P}_0-proba} 1
$$
where $\xrightarrow{\mathbb{P}_0-proba} $ is the convergence in probability under the true data-generating model. 
\end{theorem}
\noindent The proof is provided in \ref{proof:posterior_simplex}. 

\noindent Theorem \ref{thm:posterior_simplex} considers the worst-case scenario, in which all samples correspond to the same geochronological period and therefore share the same true age.
This phenomenon leads to counter-intuitive phenomenon: increasing the number of observations degrades the quality of age estimates. As sample size $n$ grows, the bias intensifies, shifting estimates further from the true ages. This is particularly problematic in practice because high-frequency sampling, which archaeologists might reasonably employ to achieve more accurate chronologies;  produces increasingly biased results under the uniform order model. Moreover, this represents a poor return on investment: dense sampling is costly in terms of both time and resources (laboratory analysis, fieldwork effort), yet rather than improving precision, it systematically distorts the chronology.

\noindent To address these issues,  \cite{nicholls_radiocarbon_2001} tackled the systematic bias inherent in the uniform order model and proposed a modified prior. 
Their approach transforms the distribution of $d_n$ into a uniform distribution over the total duration
interval $[0 ,\beta-\alpha]$:
\begin{equation}
d_n \sim \mathcal{U}(0, \beta-\alpha).
\end{equation}
Moreover, the first parameter in the order chain, conditional to duration $d_n$ is:
\begin{equation}
    \theta_1 \sim \mathcal{U}(\alpha, \beta- d_n),
\end{equation}
and $\theta_n = \theta_1 +d_n$.
Finally, conditional to first $\theta_1$ and last $\theta_n$, we have \begin{equation}
    (\theta_2, \cdots, \theta_{n-1}) \sim \mathcal{U}^*(\theta_1, \theta_n),
\end{equation}
where $\mathcal{U}^*(\theta_1, \theta_n)$ representing the uniform order distribution over the interval $\lbrack \theta_1, \theta_n \rbrack$. \\
Their correction modifies the original uniform order density \eqref{eq:uodensity} to the following expression:
\begin{equation} \label{eq:nicholls}
\pi_{\P}(\theta) = \frac{(n-2)!}{(\beta-\alpha)} \cdot \frac{1}{\bigg((\beta-\alpha) - (\theta_n-\theta_1)\bigg)(\theta_n-\theta_1)^{n-2} } \strictorder.
\end{equation}
The power of this approach stems from its hierarchical structure: we first set a duration $d_n$ before constructing the order structure thus, we are more accurate than the plain uniform 
order \eqref{eq:uodensity}. Despite its effectiveness in addressing duration bias, the prior \eqref{eq:nicholls}, for the considered archeological applications, can only address chain order structures.

\section{Adaptation of the projected posterior theory}\label{sec:projection}

\subsection{Projector for stratigraphic constraints}
Suppose an estimate $\theta' \in \R^n$ that does not include any stratigraphic knowledge. It can violate the observed constraints and $\theta' \not \in \P$. We set the projection map to satisfy the order constraints, and choose a weighted  $\ell_2$ norm objective so that samples with high posterior certainty resist movement.
\noindent The projector is defined as the following: \begin{equation} \label{eq:QP}
    T_\P(\theta') = \arg \min\limits_{\theta \in \P} \sum_{i=1}^{n} w_i (\theta'_i-\theta_i)^2,
\end{equation}
where $W$ is a positive definite and diagonal matrix. The objective function is separable and the diagonals terms are  $w_i = Var(\theta_i \mid M)^{-1} \ \forall i \in \lbrace 1, \cdots, n \rbrace $, the inverse variance under an unconstrained posterior distribution, formally defined in \eqref{eq:unconstrained}. 
On one hand, ages $\theta_i$ with high precision have large weights $w_i$, so they are heavily penalized if moved. On the other hand, ages with high uncertainty, have low weights, allowing them to adjust freely to satisfy  stratigraphic constraints. 

\noindent This is  a weighted Isotonic Regression (IR) problem, which is a convex optimization problem, with a strictly convex and coercive objective function over the convex space $\P$ (the DAG structure). The projection $T_\P$ admits a unique solution for each $\theta' \in \R^n$. 

\begin{proposition} \label{prop:block_solution}
    
The optimal solution $\theta$ can be partitioned into blocks 
\begin{equation}
    \label{eq:block_partition}
    \lbrace \theta_1, \cdots, \theta_n\rbrace = B_1 \cup \cdots \cup B_K \quad K \in \lbrace 1, \cdots, n\rbrace.
\end{equation}
Within each block, the projected chronology is constant and equals the precision-weighted average: 
\[
\theta_o = \frac{\sum_{i \in B_o} w_i \theta'_i}{\sum_{i \in B_o} w_i}. 
\]
\end{proposition}
\noindent This result is a mathematical consequence of Karush-Kuhn-Tucker (KKT) conditions. The practical insight is that optimal projection merges groups of samples and assigns each group its precision-weighted mean, respecting stratigraphic constraints. The proof can be found in \ref{proof:block_solution}.

\subsection{Theory of Posterior Projection} \label{section:theory_projection}

The section \ref{sec:background} highlighted the limitations of incorporating stratigraphic constraints directly into the Bayesian model through an ordered prior. This naturally raises the following question: can statistical inference and constraint enforcement be separated? More precisely, can one first infer an unconstrained posterior distribution from the measurement model alone, and subsequently incorporate the stratigraphic information through the projection map $T_\P$?

\noindent Starting from the unconstrained posterior
\begin{equation}\label{eq:unconstrained}
    \pi(\theta\mid M)\propto f(M\mid\theta)\pi(\theta),
\end{equation}
we define the projected posterior as the push-forward measure
\begin{equation}\label{eq:pushfoward}
    \widetilde{\Pi}(\cdot\mid M)
    =
    (T_{\mathcal P})_\#\,\Pi(\cdot\mid M).
\end{equation}
The prior $\pi$ is chosen independently of the stratigraphic
constraints, so that the unconstrained posterior is supported on the
whole parameter space $\Theta$ (e.g.  $\R^n$ or $\lbrack \alpha, \beta \rbrack^n$). Typical choices include  the improper flat prior
\begin{equation}
  \label{eq:flat}
  \pi(\theta)\propto 1_{\R^n},
\end{equation}
 or a uniform prior over a sufficiently large study period
\begin{equation} \label{eq:uniform_prior}
    \pi(\theta)
=
\frac{\mathbf 1_{[\alpha,\beta]^n}(\theta)}
{(\beta-\alpha)^n}. 
\end{equation}

The remaining question is whether the projected posterior admits a principled statistical interpretation. The answer is provided by the posterior projection framework of \citet{astfalck2026posteriorprojectioninferenceconstrained}. Their main result shows that the push-forward measure is not merely a convenient computational construction: it is the distribution supported on the constrained space that is closest to the unconstrained posterior in the Wasserstein metric. We now specialize this result to the isotonic projection $T_\P$ considered in this work.

\noindent Let $\mathcal{Q}(\R^n)$ (respectively $Q(\P)$) denote the set of probability measures with support in $\R^n$ (respectively $\P$) and finite second-order moments. 
The weighted Wasserstein-2 distance between measures $\mu, \nu \in Q(\R^n)$ is:

\begin{equation}  \label{def:wasserstein}
    \mathcal{W}(\mu, \nu) = \inf_{\gamma \in \Gamma(\mu,\nu) } \mathbb{E}_\gamma[\|X - Y\|^2_W] ^ {1/2} 
\end{equation}
where $\Gamma(\mu,\nu)$ is the set of measures over $\R^n\times \R^n$ with $\mu$, $\nu$ as marginal measures, and 
$\forall x \in \R^n \; \|x\|^2_W = x^T W x$ is the squared precision-weighted norm.

\begin{proposition}[Projected Posterior via Convex Optimization]\label{prop:projected_sampling}
Let $ \Pi(. \mid M)$ denote the unconstrained posterior measure define in \eqref{eq:unconstrained}. Suppose $\Pi(. \mid M)$ has a finite  second-order moment.
The projected posterior distribution is characterized as:
\[  \tilde{\Pi}_{\P} (. \mid M) = \underset{\nu \in \mathcal{Q}(\P)}{\arg \min }\lbrace \mathcal{W}(\nu, \Pi(. \mid M) ) \rbrace, \]
where $\mathcal{W}(\cdot, \cdot)$ is the weighted Wasserstein distance define in~\eqref{def:wasserstein}.\\
That is, $\tilde{\Pi}(. \mid M)$ is the closest distribution to the unconstrained posterior among all distribution in $\mathcal{Q}(\P)$.
\end{proposition}

\begin{proof}
    This proposition is an immediate result of the Theorem 2 in \cite{astfalck2026posteriorprojectioninferenceconstrained}, the proof is given in \ref{proof:projected_sampling}.
\end{proof}
\noindent To understand what happens when we project the unconstrained posterior, we characterize the geometry of this operation.  The normal cone plays a central role in the following developments.
\begin{definition}
    Let $\P$ be a convex subset of $\R^n$. The normal cone to $\P$ at $\theta \in \P$ is \[N(\theta, \P) = \lbrace v \in \R^n \mid \sum w_i v_i (\xi_i - \theta_i) \leq 0 \ \forall \xi \in \P \rbrace. \]
    By convention, $N(\theta, \P) = \emptyset \text{ if } \theta \notin \P$.
\end{definition}
\noindent The normal cone captures the geometry of feasibility. A key property relates the pre-image (all points that project to $\theta$) to the normal cone:
\begin{equation}
    T_{\P}^{-1}(\lbrace \theta \rbrace) = \theta +  N(\theta, \P).
\end{equation}

\noindent Geometrically, this process involves taking a point $\theta'$ and projecting it onto the nearest point within the space $\P$. Points in $\P$ remain unchanged, while points outside are projected onto $\P$ along a normal direction. This effectively reassigns the probability mass from the cone $N(\theta, \P)$ to the boundary point $\theta$, leading to a higher mass concentration at the boundary and a distortion of the original density. Using the normal cone characterization, the projected posterior can be expressed as follows 
\begin{equation}
    \widetilde{\Pi}(B\mid M)  = \Pi\left( \bigcup_{\theta \in B}
      (\theta +N(\theta, P) ) \ \mid \ M\right)  \quad \forall \ B \in  \mathcal{B}_\P, 
\end{equation}
where $\mathcal{B}_{\P}$ is the Borel $\sigma$-algebra on $\P$.

\noindent Besides, the posterior induced by the uniform prior on $\mathcal P$ in \eqref{eq:uodensity}, 
concentrates on monotone vectors with an asymptotically maximal range (see Theorem \ref{thm:posterior_simplex}),
this projected posterior displays the opposite asymptotic behavior: its interior coordinates collapse together.

\begin{theorem}
\label{thm:projected_posterior}
Consider Gaussian measurements, i.e. the likelihood function $f(\cdot \mid \theta)$ is the density of the Gaussian vector $\mathcal{N}(\theta, \sigma^2 I_n)$ in \eqref{eq:unconstrained} with the prior distribution being \eqref{eq:flat}. \\
If the true data-generating model is
\eqref{true},  then, for all \(0<a<b<1\) and for all \(\eta>0\),
the projected posterior distribution satisfies
$$
\widetilde{\Pi}\!\left(
\theta_{\lfloor bn\rfloor}-\theta_{\lfloor an\rfloor}>\eta
\,\middle|\,M
\right)
\xrightarrow{\mathbb{P}_0-proba}0,
$$
where $\xrightarrow{\mathbb{P}_0-proba} $ is the convergence in probability under the true data-generating model \eqref{true}. 
\end{theorem}
\noindent The proof can be found in \ref{proof:projected_posterior}.

\noindent An advantage of the proposed methodology lies in the interpretation of the projected posterior as a non-informative distribution with respect to the partial order constraints. The projection approach defines the constrained posterior as the Wasserstein $\mathcal{W}$ projection of the unconstrained posterior onto $\P$. Explicitly, it identifies, among all probability measures satisfying the required ordering, the one that is closest to the unconstrained posterior  for the $ \mathcal{W}$ distance.
This contrasts with classical archaeological Bayesian models based on priors of the form \eqref{eq:constrainedprior}. For instance, the choice of the prior \eqref{eq:uodensity} induces non-negligible distortions, often leading to the well-known duration-inflation bias as explained in section \ref{sec:background}.

\subsection{Synthetic two-dimensional example}
\noindent To gain intuition about the geometry of the projected posterior, we study a simple two-dimensional example for which the projected posterior density can be derived in closed form.

    \label{dummy_example}
\noindent Consider an unconstrained Gaussian prior with known standard-deviations $\tau$:
\[ \mathcal{N}(\mu ,  \tau^2 I_2),\] and a Gaussian likelihood with the known standard deviation $\sigma$:  \[M \sim \mathcal{N}(\theta, \sigma^2I_2).\]
The unconstrained posterior  $\Pi(\cdot \mid M)$ is the Gaussian distribution $\mathcal{N}_2(\mu_{M} ,  \sigma_M^2 I_2),$ where
 \begin{equation*}
    \mu_{M_i} = \frac{M_i \tau^2 + \mu \sigma^2}{\sigma^2 + \tau^2 } \ \text{; }  i \in \lbrace 1,2\rbrace \quad
        \sigma_M^2 = \frac{\sigma^2 \tau^2}{\sigma^2 + \tau^2}.
 \end{equation*}
The constrained parameter space is  $\P = \lbrace (x_1, x_2) \in \R^2 \mid x_1 \leq x_2 \rbrace$.\\
This implies that samples $\theta' \sim \Pi(\cdot \mid M)$ located above the diagonal $D = \lbrace \begin{pmatrix}
        x \\ y
    \end{pmatrix} \in \R^2 \mid x = y \rbrace,$ remain unchanged, while those below are projected onto the diagonal. The projection map admits the explicit form
\[\theta = T_\P(\theta') = \begin{cases}
     \theta' & \text{ if } \theta'_1 \leq \theta'_2 \\
    \frac{\theta'_1 + \theta'_2}{2}  \begin{pmatrix}
        1 \\ 1
    \end{pmatrix}  &\text{ if } \theta'_1 > \theta'_2.
\end{cases}
\]
The projector $T_\P$ is continuous over $\R^2$, although it is not differentiable on the boundary of $D$, where the active set changes. Applying the definition \eqref{eq:pushfoward} yields the following Radon-Nikodym density of the projected posterior (derivation is provided in \ref{proof:dummy_example}).
\begin{equation*}
    \tilde{\pi}_{\P}(\theta \mid M) \propto \phi(\theta_1; \mu_{M_1}, \sigma_M ) \phi(\theta_2; \mu_{M_2}, \sigma_M )  \1_{\theta_1 < \theta_2} + g_M(\theta)\1_{\theta_1 = \theta_2},
\end{equation*}
with respect to the measure: \[\lambda(B) = \mu_\R(D \cap B) + \mu_{\R^2}(\P \cap B), \] and where
\begin{equation*}
    g_M(\theta) \propto  \phi\bigg(\theta;\frac{\mu_{M_1}+\mu_{M_2}}{2}, \frac{\sigma_M}{\sqrt{2}} \bigg)\times  \Phi\bigg(\frac{\mu_{M_1}-\mu_{M_2}}{\sqrt{2}\sigma_M}\bigg).
\end{equation*}
The function $\phi(.; m, \sigma^2)$ denotes the density of the Gaussian distribution  $\mathcal{N}(m ,  \sigma^2)$, $\Phi$, the cumulative distribution function of $\mathcal{N}(0 ,   1)$.

\begin{remark}
    The density expression naturally decomposes into two components. Inside the constrained space $\P$, the density coincides with the unconstrained Gaussian posterior density. By contrast, probability mass associated with infeasible samples is transported onto the boundary $D$. Consequently, the projected posterior contains an absolutely continuous component on the interior of $\P$ together with a singular component supported on the diagonal leading to a change of measure (denoted $\lambda$). 
\end{remark}

\noindent This simple example provides insight into the behavior of the projected posterior in situations where the ordering constraint conflicts with the measurements. We therefore consider two neighbouring events with true ages $(10, 11)$, and focus on the case where measures are reversed, $M_1 > M_2$. This scenario is illustrated in Figure \ref{fig:DummyexampleProj}. 

\begin{figure}
    \centering
    \includegraphics[width = .9\linewidth]{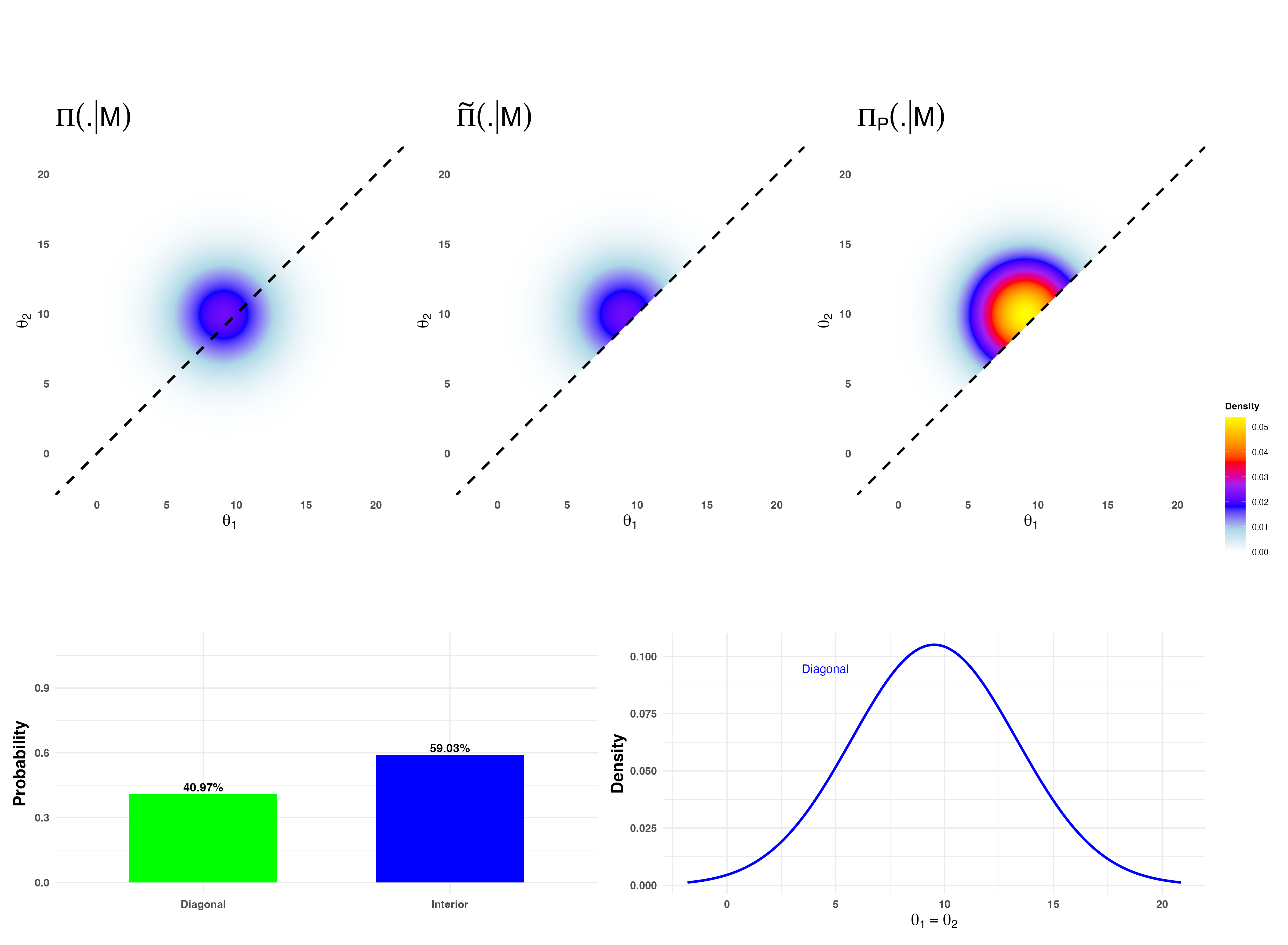}
    \caption{Two-dimensional Gaussian example. Top row: unconstrained posterior (left), projected posterior onto the isotonic cone $\P$ (middle), and constrained Bayesian posterior (right). Since the projected posterior assigns positive probability mass to the diagonal $D$, this singular component is displayed separately in the bottom-right panel. The bottom-left panel reports the posterior probabilities of lying inside and outside the admissible region $\P$.}
    \label{fig:DummyexampleProj}
\end{figure}
\noindent In the top-left panel of Figure \ref{fig:DummyexampleProj}, one can observe that a substantial proportion of the posterior mass lies outside the constrained space~$\P$, with high-density regions appearing below the diagonal $D$. The corresponding probability of violating the order constraint is summarized in the bottom-left bar plot. \\
It is helpful to compare this projected-posterior approach with the alternative strategy of imposing the constraints directly within the Bayesian model, i.e., by working with the constrained prior \eqref{eq:uodensity}. \\
This comparison highlights the differences between projecting an unconstrained posterior and sampling directly under stratigraphic order constraints included in the prior.
\noindent In the constrained Bayesian approach, infeasible configurations are removed before posterior inference through the ordered prior. The remaining probability mass is therefore renormalized over the admissible region, which favours chronologies exhibiting larger temporal separation between neighboring ages. Consequently, the posterior is shifted away from the boundary of the cone. This shift directly translates into inflated duration estimates.
This effect becomes particularly apparent in the marginal posterior distributions (Figure \ref{fig:marginals}).
\begin{figure}
    \centering
    \includegraphics[width = .9\linewidth]{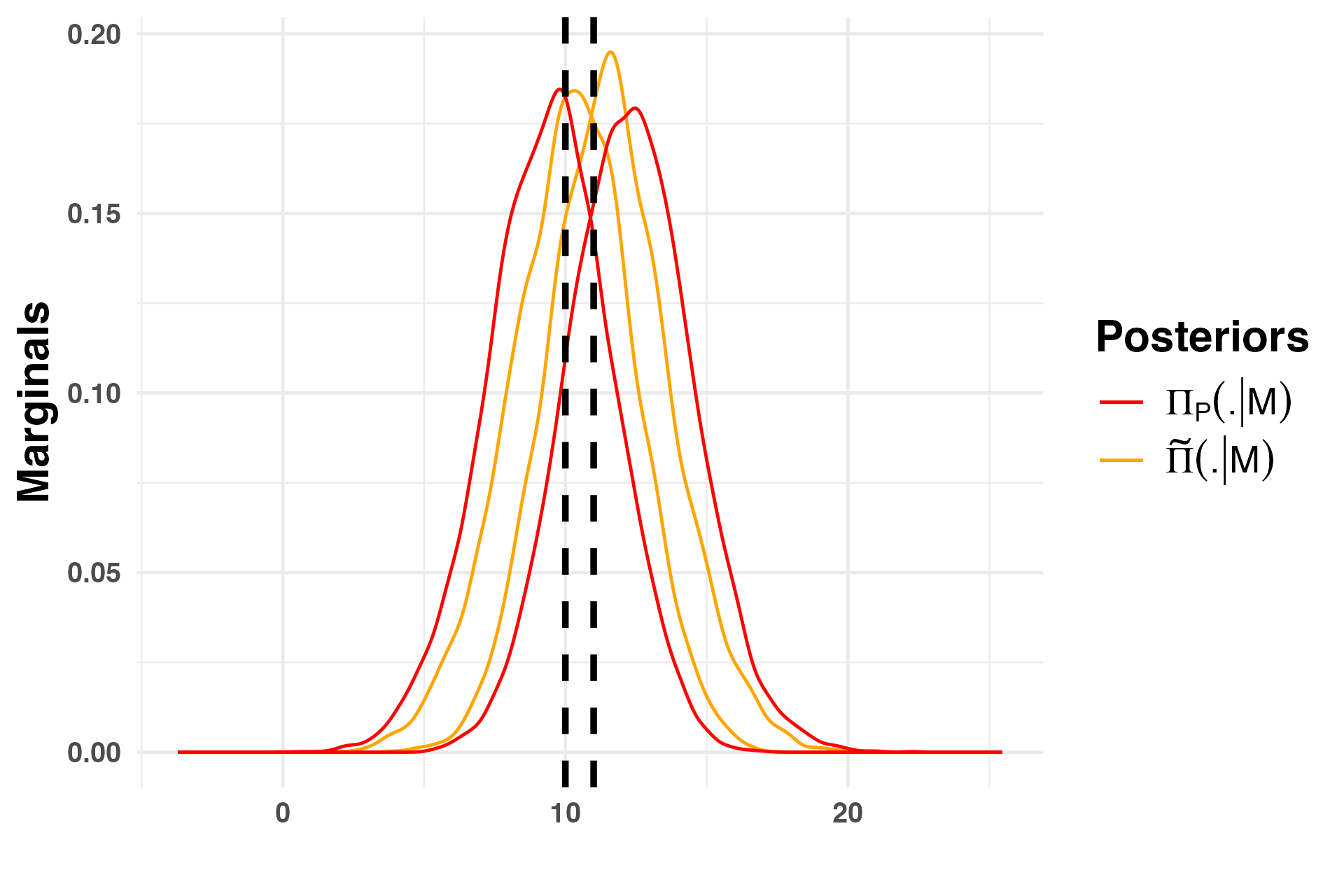}
    \caption{Posterior marginals under the projected and constrained approaches. Dashed lines correspond to the true ages $\theta_1 = 10$ and $\theta_2 = 11$.}
    \label{fig:marginals}
\end{figure}

\subsection{Computational approach}

The projected-posterior framework naturally leads to a two-stage computational procedure. First, samples are generated from the unconstrained posterior distribution using any standard MCMC algorithm. Second, each posterior draw is projected onto the isotonic cone by solving the QP \eqref{eq:QP}. Consequently, posterior inference reduces to repeatedly solving a convex optimization problem.

\noindent The computational bottleneck therefore lies in the projection step.  We instead propose the dedicated algorithm described in Algorithm~\ref{alg:projected_posterior_sampling}, which combines posterior sampling with an efficient isotonic projection solver.

\begin{algorithm}
\caption{Projected Posterior Sampling via Isotonic Regression}
\label{alg:projected_posterior_sampling}
\begin{algorithmic}[1]
\REQUIRE Unconstrained posterior $\Pi(\cdot \mid M)$; 
         constraint set $\P$ encoded as DAG $G = (V, E)$;
         weights $w_i = \mathbb{V}(\theta_i \mid M)^{-1}$ for $i \in V$;
         number of samples $p$
 
\ENSURE Samples $\{\theta^{(1)}, \ldots, \theta^{(p)}\} \sim \widetilde{\Pi}(\cdot \mid M)$
 
\STATE \textbf{Step 1: Unconstrained Sampling}
\STATE Draw $p$ samples from the unconstrained posterior:
\STATE $\quad \{\theta'\}^{(1)}, \ldots, \{\theta'\}^{(p)} \sim \Pi(\cdot \mid M)$
 
\STATE \textbf{Step 2: Constraint Projection} 
\FOR{$j = 1$ to $p$}
    \STATE Solve the weighted isotonic regression problem:
    \STATE $\quad \theta^{(j)} \gets \arg\min_{\xi \in \P} \sum_{i=1}^{n} w_i (\xi_i - \{\theta'\}^{(j)}_i)^2$
    \STATE \quad (Using Sequential Block Merging algorithm)
\ENDFOR
 
\STATE \textbf{Step 3: Output}
\RETURN $\{\theta^{(1)}, \ldots, \theta^{(p)}\}$ as samples from $\widetilde{\Pi}(\cdot \mid M)$
\end{algorithmic}
\end{algorithm}

\noindent The isotonic regression problem \eqref{eq:QP} can be solved by several methods, each offering different computational complexities and theoretical guarantees. A first approach consists of using Interior Point Methods (IPMs) through existing convex optimization solvers, such as those implemented in the CVX software. Although IPMs provide polynomial-time convergence guarantees (\cite{kyng2015fast}), their computational cost is dominated by the repeated solution of large KKT systems, making them unsuitable for large-scale problems. Alternative approaches, such as Active Set Methods (ASMs), have also been proposed (\cite{leeuw_isotone_2009}). While these methods can accommodate objective functions beyond quadratic losses, they have been shown to be computationally expensive.

The block-wise structure can instead be exploited through dedicated
algorithms specifically designed for isotonic regression; we refer to
these as structured algorithms. The classical Pool Adjacent Violators
Algorithm (PAVA; see \citealp{best_active_1990}, and references therein) is the canonical example. It computes the exact solution in linear time, $O(n)$, for total (chain) orders. However, it is not directly applicable to general partial orders.

\noindent Several structured algorithms have subsequently been proposed for the partial-order setting. Among them, \cite{burdakov_on2_2006} introduced the Generalized PAV (GPAV) algorithm. Benchmark studies demonstrated its computational efficiency, with a complexity of $O(n^2)$ and good numerical accuracy. Nevertheless, the returned solution remains an approximation of the optimal solution.

\noindent In this work, we adopt the Sequential Block Merging (SBM) algorithm proposed by \cite{wang_efficient_2022}. It guarantees a finite-step convergent algorithm for general isotonic optimization and exhibits substantially faster execution in benchmark studies. Furthermore, it naturally extends to a broader class of separable objective functions.

\section{Simulation Study: Revisiting the Steier–Rom Experiments}
\label{sec:simulations}
As mentioned earlier, \cite{steier_use_2000} conducted simulations to highlight the spread problem induced by the constrained Bayesian approach. The setup was straightforward: simulated measurements were generated from an ordered sequence of ages $\P = \lbrace \theta_1 \leq \cdots \leq \theta_n \rbrace$, where each consecutive pair is separated by a fixed interval $\Delta_t = \theta_{i+1} - \theta_i$.
    \begin{equation}
        \theta_k = 1000 + \left[ k - \frac{(n+1)}{2} \right] \Delta_t.
    \end{equation}
To obtain measures, we proceed as follows: 
\begin{equation}\label{eq:simulated_measures}
     M = \begin{bmatrix} M_1 \\ \vdots \\ M_n \end{bmatrix} = \mathbf{\varepsilon} + 
     \begin{bmatrix} \theta_1 \\ \vdots \\ \theta_n \end{bmatrix} \sim \mathcal{N}\left( \begin{bmatrix} \theta_1 \\ \vdots \\ \theta_n \end{bmatrix}, \sigma^2 I_n \right),
\end{equation}
where \[\mathbf{\varepsilon} \sim \mathcal{N}(0_n, \sigma^2 I_n). \]
In order to compare with the benchmark of \cite{steier_use_2000}, the standard deviation is fixed at $\sigma = 100$. For each simulated dataset, we compare four inference strategies:  
\begin{itemize}
\item  Reference model 
\begin{itemize}
    \item Unconstrained model:
      $\Pi(\cdot\mid M)$ obtained with a flat prior.
\end{itemize}

\item Classical Bayesian approaches
\begin{itemize}    
    \item Uniform-order model:
          $\Pi_{\mathcal P}(\cdot\mid M)$ using the classical uniform-order prior \eqref{eq:uodensity}.
    
    \item Corrected uniform-order model:
          $\Pi_{\mathcal P}^{N}(\cdot\mid M)$ using the duration correction \eqref{eq:nicholls} proposed by \citet{nicholls_radiocarbon_2001}.
\end{itemize}

\item Proposed approach
\begin{itemize}
    \item Projected-posterior model:
          $\widetilde{\Pi}(\cdot\mid M)$, corresponding to the proposed projection framework.
\end{itemize}
\end{itemize}

\subsection{Experiment A: Varying Age Gaps.}
The first experiment, originally proposed by \citet{steier_use_2000}, investigates the effect of the temporal separation between successive events. The age gap
\[
\Delta_t=\theta_{i+1}-\theta_i
\]
is progressively increased while the measurement uncertainty remains fixed. When $\Delta_t \ll \sigma$, consecutive true ages are poorly separated relative to the dating uncertainty, making inversions of the observed measurements likely. Specifically,
\[
\mathbb{P}(M_{i+1}<M_i)
=
\Phi\!\left(-\frac{\Delta_t}{\sqrt{2}\sigma}\right),
\]
where $\Phi$ denotes the standard normal cumulative distribution function. Consequently, small values of $\Delta_t$ correspond to the regime in which order constraints are expected to have the strongest influence on posterior inference.

\noindent Figure~\ref{fig:expA} displays Bayes estimate's boxplot for increasing values of $\Delta_t$. 

\noindent For small values of $\Delta_t$ ($\Delta_t = 0, \ 20$), the Corrected duration $\Pi_\P^N(\cdot \mid M)$ and projected $\widetilde{\Pi}(\cdot \mid M)$ models perform better at capturing the true ages within their HPD regions, while the Uniform-order $\Pi_\P(\cdot \mid M)$ model struggles with the age closeness. As $\Delta_t$ increases ($\Delta_t = 100$), the $\Pi_\P(\cdot \mid M)$ model gradually converges toward the true ages and provides better coverage.

\begin{figure}
    \centering
    \includegraphics[width = .9\linewidth]{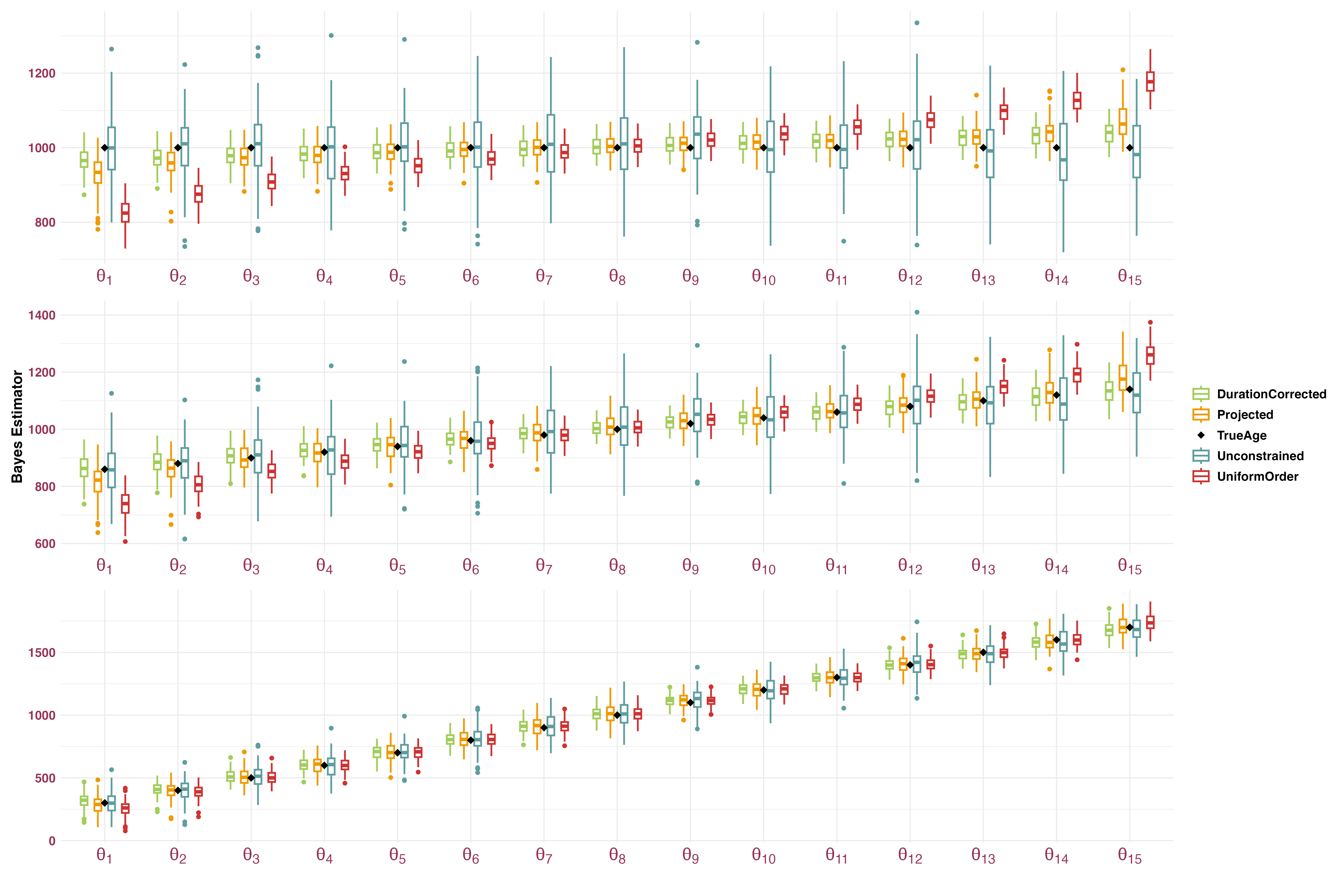}
    \caption{Boxplot of Bayes estimates obtained from $100$ replications of the simulated measures $M$ \eqref{eq:simulated_measures} when applying Experiment A with sample size $n=15$. Top: $\Delta_t = 0$. Middle: $\Delta_t = 20$. Bottom: $\Delta_t = 100$} 
    \label{fig:expA}
\end{figure}

\noindent The projected posterior $\widetilde{\Pi}(\cdot \mid M)$, exhibits a weaker performance at the boundaries of the ordered sequence, mostly for the first and last ages. This is the price paid for adopting a least-informative strategy: unlike the constrained models, the projection does not induce additional information beyond the observed measurements. Since the boundary ages are constrained on only one side, less information is available to refine their estimates. Consequently, the estimates $\theta_1$ and $\theta_{15}$ are generally less accurate than those obtained with the Duration corrected posterior $\Pi_\P^N(\cdot \mid M)$. 

\subsection{Experiment B: Increasing sample size with identical ages.}

The second experiment, also introduced by \citet{steier_use_2000}, investigates the effect of increasing the number of dated samples while all events belong to the same geochronological period. Accordingly, we consider the extreme configuration in which
\[
\Delta_t = 0,
\]
so that all true ages are identical. This setting represents the strongest possible conflict between the data-generating process and the ordered prior: although the stratigraphic order must be respected, the true chronology exhibits no temporal separation. As the number of samples increases, random measurement fluctuations produce an increasing number of apparent inversions, making this scenario particularly demanding for constrained chronological models. This corresponds to the worst-case scenario identified in Theorem~\ref{thm:posterior_simplex} and Theorem~\ref{thm:projected_posterior}, where the duration inflation induced by order constrained priors is expected to be most pronounced.

\noindent Figure ~\ref{fig:expB}  illustrates the performance of each model as sample size increases. The Uniform-order model clearly fails to capture the closeness of the ages: as $n$ increases, the posterior means spread increasingly across the age range, contradicting the knowledge that all ages are identical. In contrast, the Corrected duration and Projected models maintain tight HPD regions around the true common age, even when faced with measurement inversions. This demonstrates the efficiency of these methods for handling close ages with inverted measurements.

\begin{figure}
    \centering
    \includegraphics[width = \linewidth]{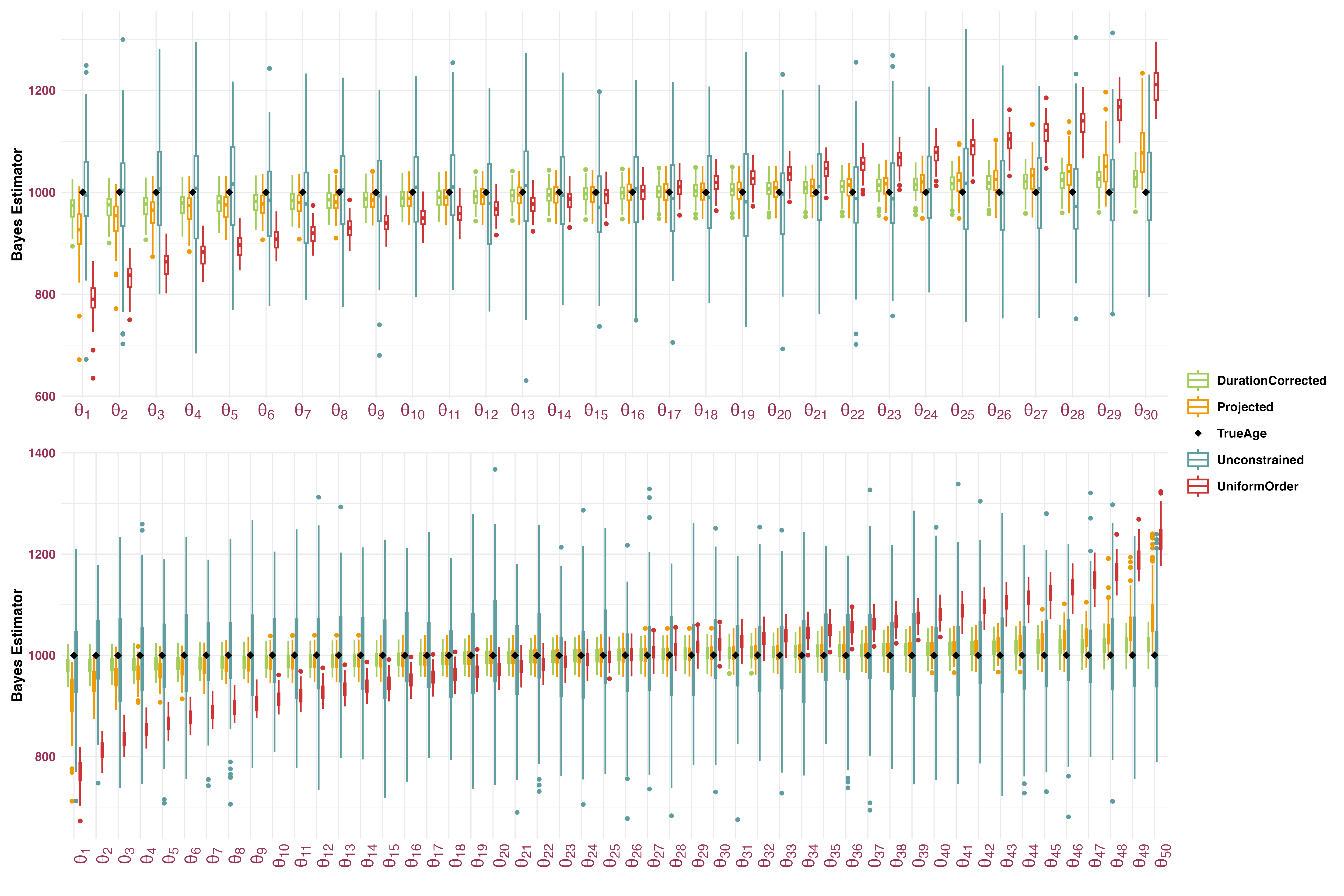}
    \caption{Boxplot of Bayes estimates obtained from $100$ replications of the simulated measures $M$ \eqref{eq:simulated_measures} when applying Experiment B with samples of size $n=30$ and $n=50$.}
    \label{fig:expB}
\end{figure}

\noindent The poor performance of the Uniform-order model becomes even more pronounced at larger sample sizes. Figure~\ref{fig:expB} shows that with $n=50$, the $\Pi_\P(\cdot \mid M)$ distribution's Bayes estimate spans nearly the entire feasible age range, providing essentially no useful information about the true ages. This systematic bias confirms the findings of \cite{steier_use_2000} regarding the failure of uniform order priors under dense sampling scenarios.

\noindent The runtimes obtained under Experiment~B for sample sizes $n=30$ and $n=50$ are reported in Figure~\ref{fig:runtime}. Computational time required by each model reveals a clear advantage for the Projected one. It should be noted that the projected posterior runtime is the required time for the QP solver (SBM) plus the unconstrained runtime. Furthermore, all constrained Bayesian models are sampled using Just Another Gibbs Sampler (JAGS), a black box not necessarily optimized for truncated distributions. 
We apply the SBM sequentially to each MCMC sample, consequently, the reported runtimes are conservative, and additional computational gains could be achieved through parallel implementation.

\begin{figure}
    \centering 
    \includegraphics[width = .9\linewidth]{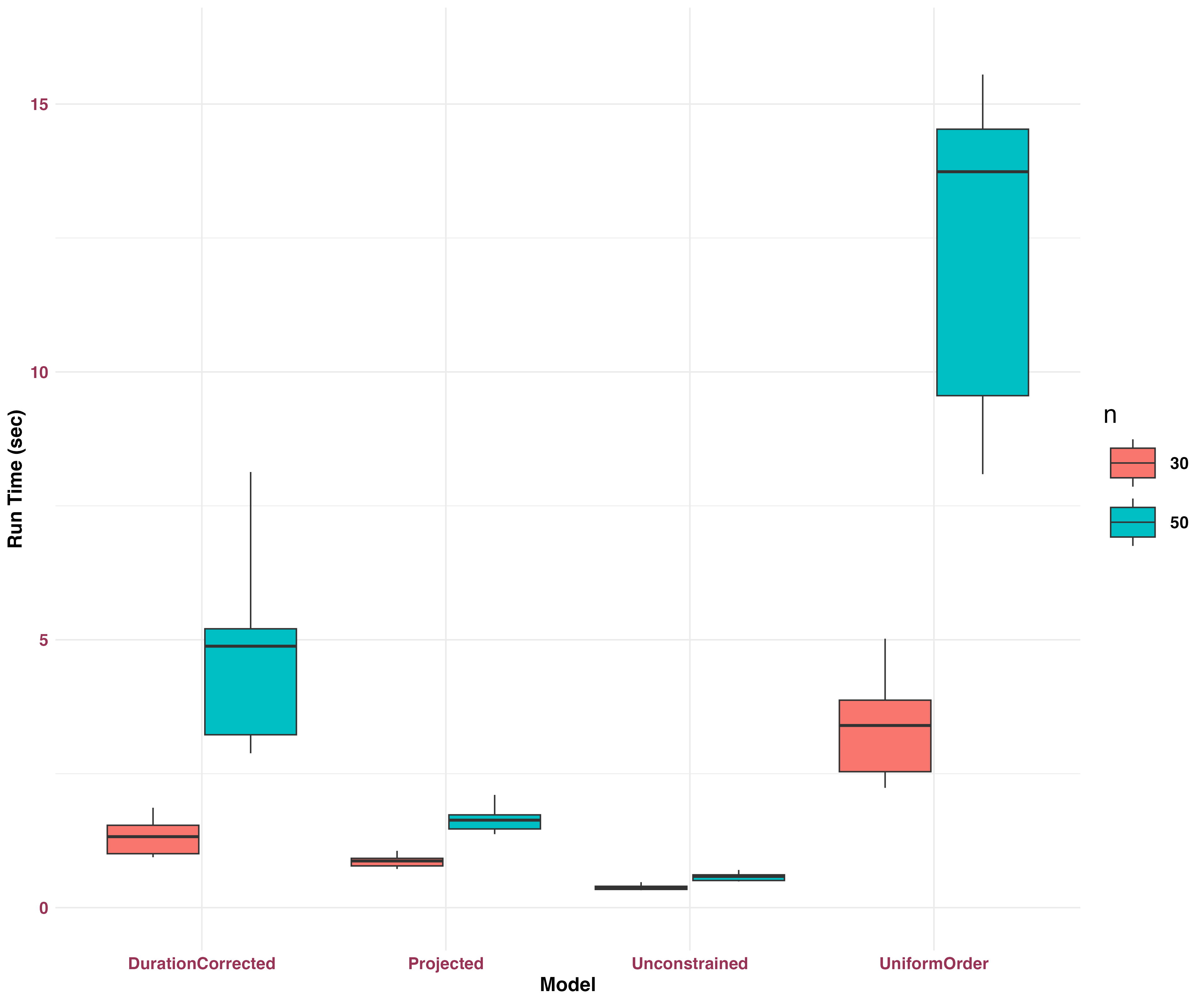}
    \caption{Computational time (seconds) over 100 simulation replicates for the four competing models under Experiment~B with sample sizes $n=30$ (left) and $n=50$ (right). Runtimes were measured using Sys.time. All computations were performed in R 4.5.0 on a MacBook Air (Mac15,12) equipped with an Apple M3 chip (8 cores) and 16 GB of unified memory, running macOS 14.6.}
    \label{fig:runtime}
\end{figure}

\subsection{Experiment C: Coverage and Accuracy Assessment} 
In this experiment, only a pair of ages $(\theta_1, \theta_2)$ is considered and multiple measurements $M^{(j)} = (M_1^{(j)}, M_2^{(j)})$ are generated for all $j \in \{1, \ldots, p\}$. For each case, in \cite{steier_use_2000}, it is examined whether the true ages fall within the 95\% HPD regions obtained from the Bayesian models. Specifically, for each posterior distribution's marginal of each model, the lower and upper bounds $(l_M, U_M)$ are computed such that
\[\mathbb{P}(\theta_i\in (l_M, U_M) \mid M) = 0.95.\]
Four quantities are then evaluated for each age $\theta_i$, where $i \in \{1,2\}$:
\begin{align*}
    \text{miss\_rate}_i &= \frac{1}{p} \sum_{j=1}^{p} (1-\1_{\theta_i^{\text{true}} \in (l_M^{(j)}, U_M^{(j)} ) }) \\
    \text{mean\_absolute\_bias}_i &= \frac{1}{p} \sum_{j=1}^{p} \mid \widehat{\mathbb{E}}(\theta_i \mid M_i^{(j)}) - \theta_i^{\text{true}} \mid \\
    \text{MSE}_i &= \frac{1}{p} \sum_{j=1}^{p} \bigg\lbrack \theta_i^{\text{true}} -\widehat{\mathbb{E}}(\theta_i \mid M_i^{(j)})\bigg\rbrack^2 \\
    \text{mean\_variance}_i &= \frac{1}{p} \sum_{j=1}^{p} \widehat{Var}(\theta_i \mid M_i^{(j)}), 
\end{align*}
where $\widehat{\mathbb{E}}$ and $\widehat{Var}$ are the Monte Carlo estimates of the expectation and variance under the posterior distribution.\\  
\noindent The miss rate quantifies coverage (expected 5\% for nominal 95\% HPD regions when Bayesian and frequentist cases are similar), while the average absolute bias, variance, and MSE assess point estimation accuracy. The goal is to determine whether models that provide great precision (narrower intervals) maintain proper coverage and accuracy.
We repeat this analysis across several values of the age gap $\Delta_t$ to examine performance as a function of age separation. Figures~\ref{fig:ExCA1} and \ref{fig:ExCA2} present the results for the youngest age $\theta_1$ and oldest age $\theta_2$, respectively.

 \begin{figure}
        \centering
        \includegraphics[width = \linewidth]{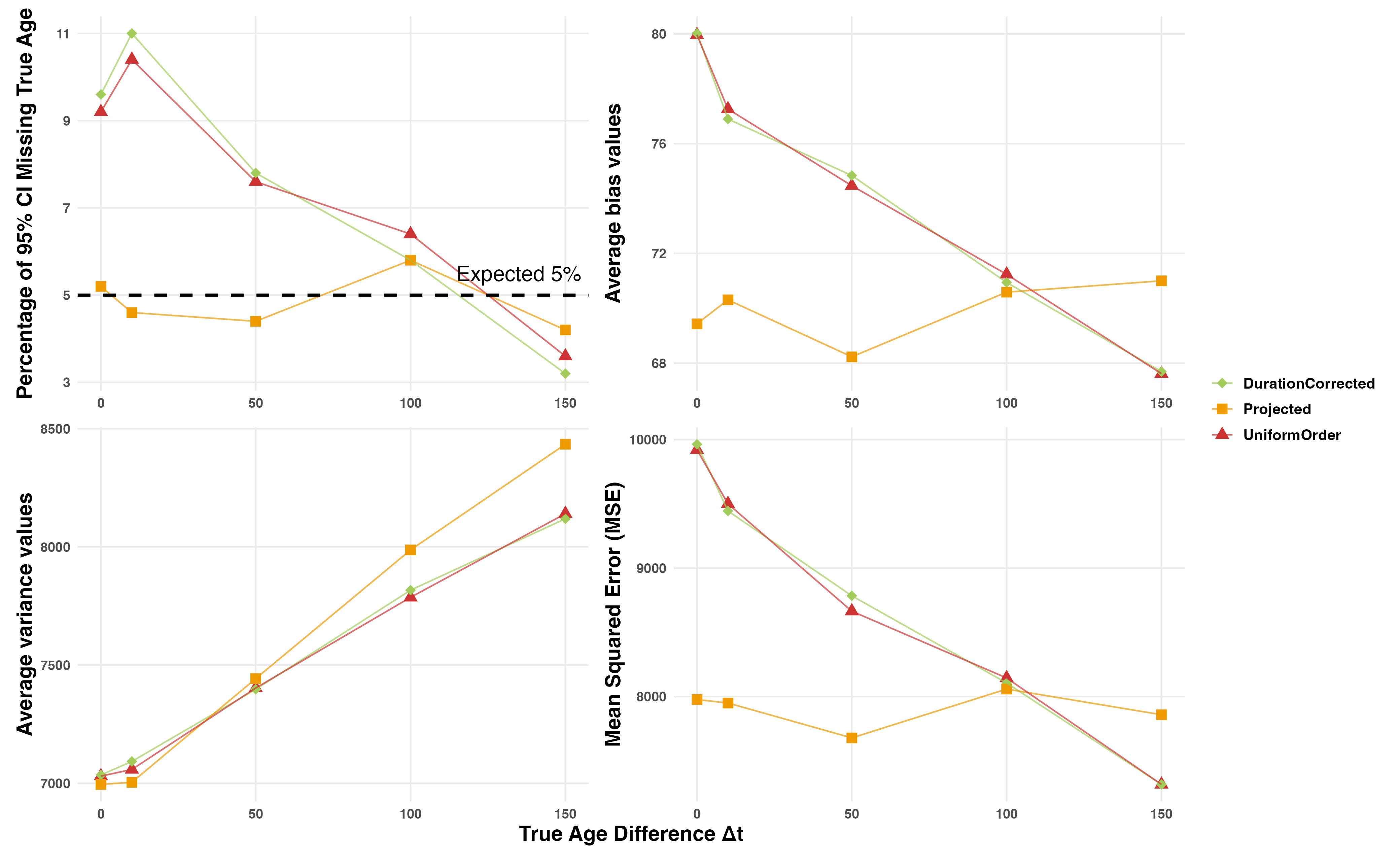}
        \caption{Experiment C results for the youngest age $\theta_1$ based on $p=500$ sampled measurement pairs. 
        The four panels display: mean bias (top right), mean variance (top left), percentage of predictions missing the true value (bottom left), and mean squared error (bottom right).}
        \label{fig:ExCA1}
    \end{figure}

 \begin{figure}        
    \centering
        \includegraphics[width = .95\linewidth]{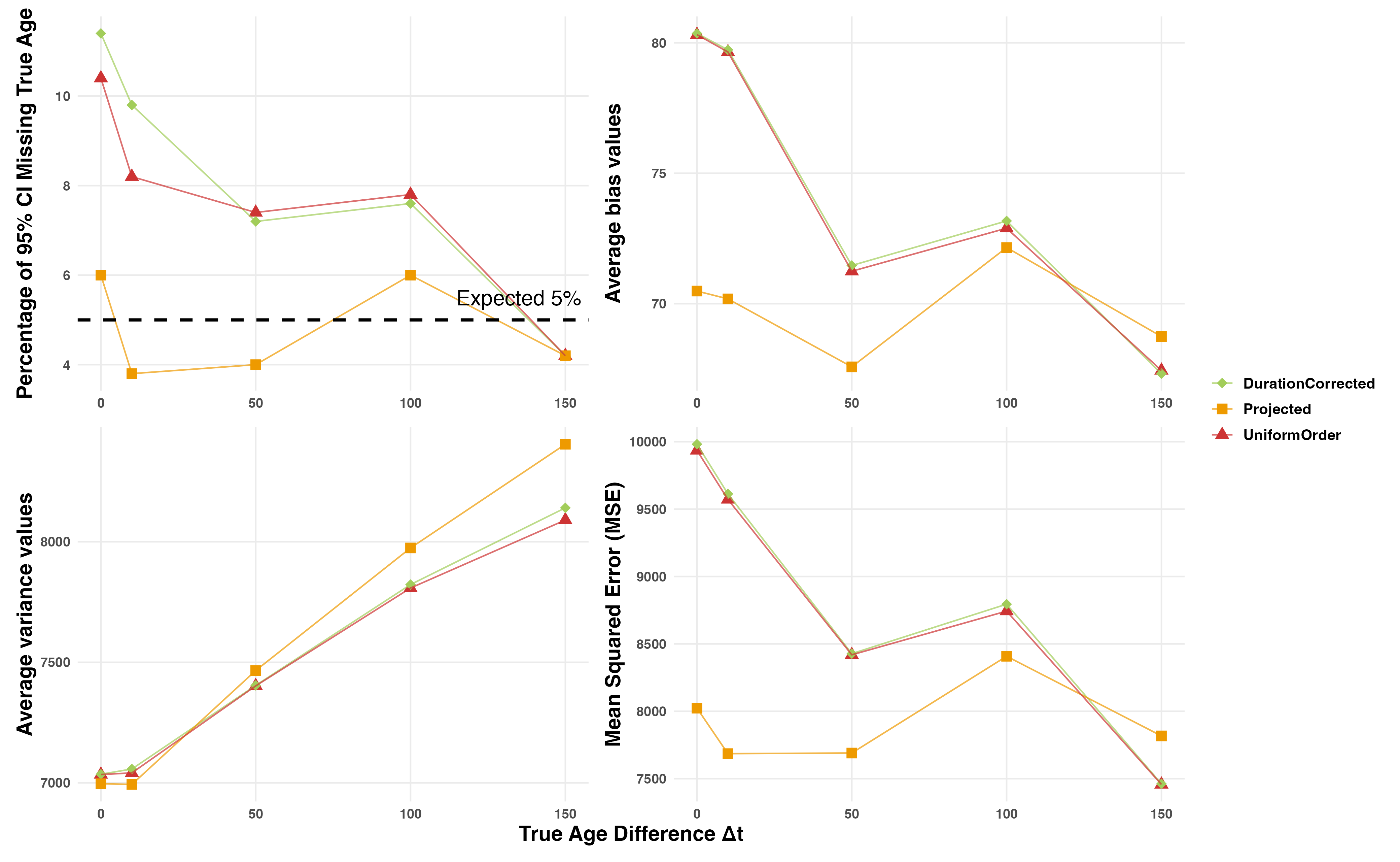}
        \caption{Experiment C results for the oldest age $\theta_2$ based on $p=500$ sampled measurement pairs. 
        The four panels display: mean bias (top right), mean variance (top left), percentage of predictions missing the true value (bottom left), and mean squared error (bottom right).}
        \label{fig:ExCA2}
    \end{figure}
\noindent The results reveal distinct performance regimes. When $\Delta_t \ll \sigma$, the Projected model demonstrates superior performance for both ages. Most notably, its miss rate remains closest to the nominal 5\% level, indicating proper coverage, while simultaneously achieving lower bias and MSE compared to competing methods. This suggests that the posterior $\widetilde{\Pi}(\cdot \mid M)$ approach successfully balances precision with accuracy in the challenging regime of closely spaced ages. These findings corroborate the intuition gained from the two-dimensional example (Figure~\ref{fig:marginals}): when neighboring ages are poorly separated relative to the dating uncertainty, posterior projection preserves the measurement evidence more faithfully than the ordered-prior approach.
\noindent However, when $\Delta_t \geq \sigma$, the Projected model's advantage diminishes. Its performance becomes comparable to other methods, with slightly higher variance suggesting some degree of overspreading in the posterior predictions. This is expected: when ages are well-separated relative to measurement error, the constraint information becomes less informative, and the projection step adds marginal uncertainty without substantially improving point estimates.

\section{Applications to archaeological case studies}
\label{sec:applications}
\subsection{Neolithic urbanization in Anatolia}

The Neolithic is a key period in human evolution, as it saw major changes in lifestyle. Among these, one of the most spectacular is the appearance of settled life, in conjunction with the onset of agriculture and animal husbandry. Therefore, understanding the timing and rate of appearance of these practices is of great importance. In this context, and outside the Fertile Crescent, two sites located in central Anatolia play a major role: Boncuklu \citep{baird2012boncuklu} was first inhabited by sedentary groups who complemented foraging practices by consuming low amounts of cereals and legumes; later, larger groups - whose material culture seems to have evolved from that found in Boncuklu - settled in Çatalhöyük East \citep{hodder2014ccatalhoyuk}, where agricultural practices are much more developed and where houses are not only more numerous but also found closer to each other compared to Boncuklu \citep{asouti2013contextual}. The central chronological question is to determine when permanent occupation began at the site, and thus how much time elapsed between the occupations of Boncuklu and Çatalhöyük East.

In Boncuklu, based on 9 radiocarbon ages, the occupations ended between 7950 and 7710 BC (although it should be noted that the youngest dated sample has an age between 7950 and 7760 BC; both intervals given at the 95$\%$ credible level, suggesting that the end of occupations as determined by OxCal implies final occupations after the youngest dated sample; in other words, OxCal extends occupations outside of the measurement range). This section presents a re-analysis of the much larger radiocarbon dataset from  Çatalhöyük East, originally analysed in \cite{bayliss_getting_2015}. In the latter publication, a model built with OxCal dated the onset of occupations from 'close to 7100 BC, at least 200 years later than previous estimates. However, in a recent article \cite{guerin_conflict_2026} showed that (i) the onset of occupations as estimated with OxCal is younger than the data suggest; and (ii) no model based on the uniform order prior is able to provide estimates that are consistent with the data.
  
\noindent The Çatalhöyük East dataset comprises 40 radiocarbon ages obtained from 40 samples, together with a stratigraphic sequence that provides a partial order DAG (see Figure \ref{fig:DagCatalhoyuk}). The individual calibration of radiocarbon measurements has the following unconstrained posterior density: 
\begin{equation*}
    \pi(\theta\mid M) =  \prod_{i=1}^{n} \pi(\theta_i\mid M_i) \propto \prod_{i=1}^{n} \frac{\exp\bigg\lbrack \frac{-(g(\theta_i)- M_i)^2}{2(\sigma^2 + \sigma_{g}^2(\theta_i))}\bigg\rbrack}{\sqrt{\sigma^2 + \sigma_{g}^2(\theta_i)}} ,
\end{equation*}
where $g, \sigma_g$ are the calibration IntCal20 curve and errors (see \cite{reimer_intcal20_2020}), and $\sigma$, the measurement errors.  \\
This constitutes the unconstrained model and is projected into the constrained space $\P$, induced by the DAG in Figure~\ref{fig:DagCatalhoyuk}.
\noindent In \cite{bayliss_getting_2015}, the chronological analysis was performed using the OxCal framework \citep{bronk_ramsey_bayesian_2009}. Figure~\ref{fig:Catalhoyuk_comparaison} compares the resulting chronology with those obtained from the Unconstrained and the Projected models.

\noindent The Projected and OxCal models lead to noticeably different posterior intervals. In particular, the OxCal credible intervals exhibit a marked concentration effect. This behavior arises from the phase model incorporated in OxCal, where stratigraphic information is introduced directly into the Bayesian model through additional parameters, corresponding to phase boundaries. Consequently, with this prior information on phase parameters,  the posterior distribution becomes more concentrated than under the prior \eqref{eq:uodensity} (see \cite{lanos_event_2018} for a detailed discussion).
A striking example is provided by sample OxA-9776. Under the OxCal posterior, its credible interval is substantially shifted and shows little agreement with the corresponding interval obtained from the Unconstrained model. As for the Projected values, they are  consistent with the calibrated ages, showing good agreement between the isotonic estimates, the stratigraphic order, and previous expectations for the onset of occupation at Çatalhöyük East. 

   \begin{figure}
    \centering
    \includegraphics[width = \linewidth]{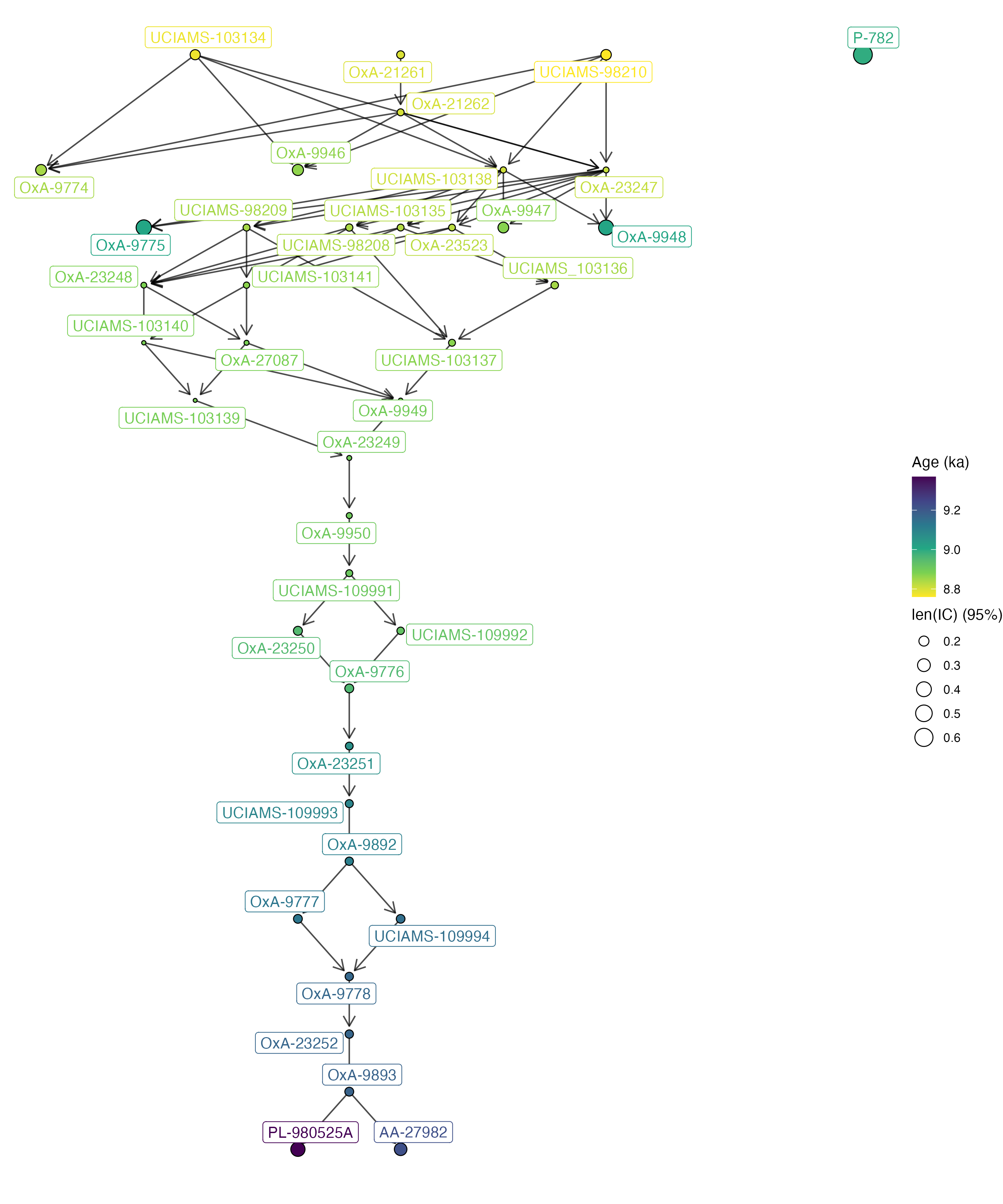}
    \caption{Directed acyclic graph (DAG) of the stratigraphic constraints for the Çatalhöyük dataset. Node colors represent the unconstrained posterior mean ages, and node sizes are proportional to the longer of the corresponding 95\% credible intervals.}
    \label{fig:DagCatalhoyuk}
   \end{figure}

   \begin{figure}
    \centering
    \includegraphics[width = \linewidth]{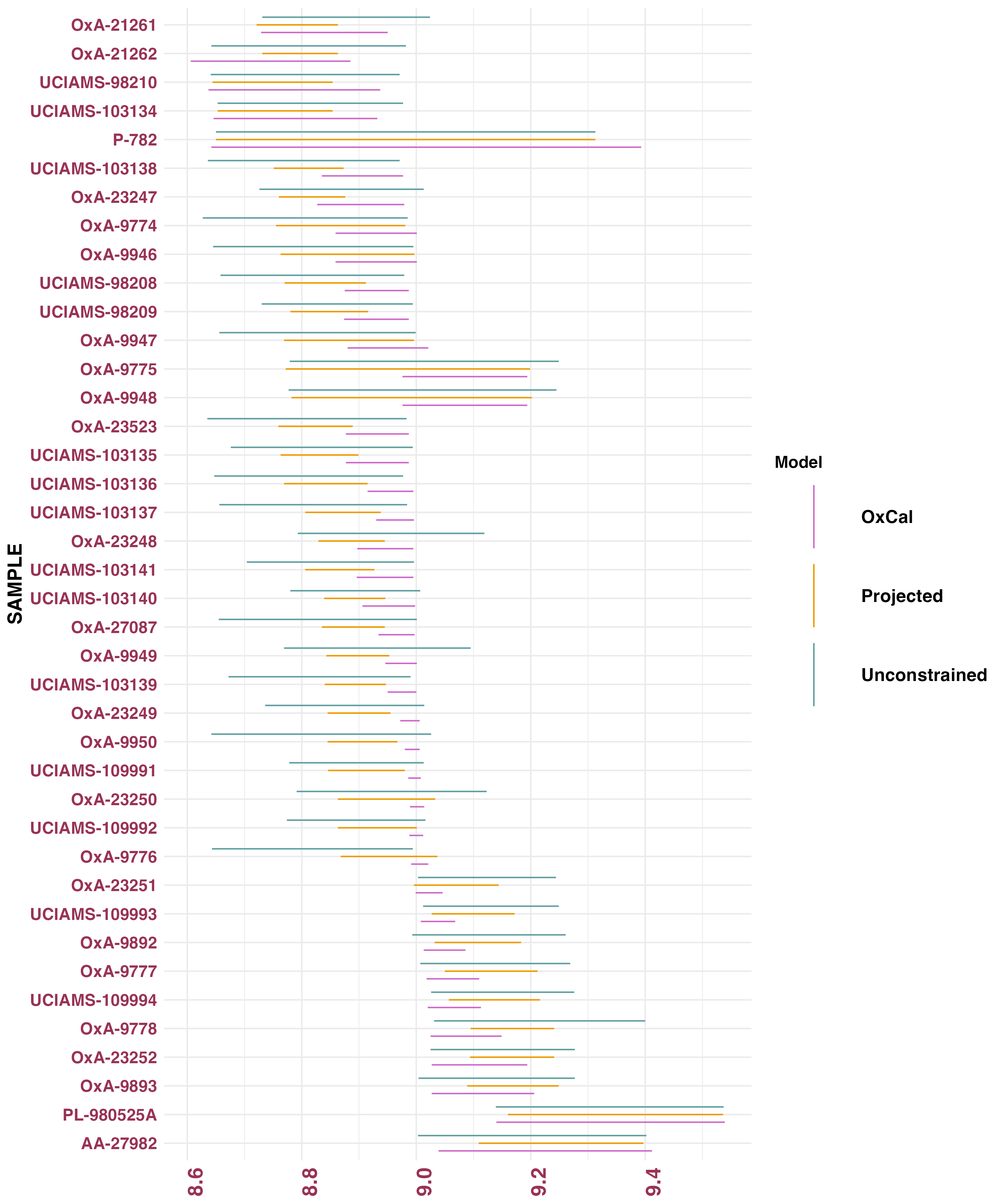}
    \caption{Comparison of the 95\% credible intervals obtained with the OxCal, Projected, and Unconstrained models.}
    \label{fig:Catalhoyuk_comparaison}
   \end{figure}

\subsection{The Middle Palaeolithic of Gatzarria in its climatic context}

Located in Ossas-Suhare in western Pyrenees (Pyrénées-Atlantiques, south west France), Gatzarria Cave has yielded numerous lithic artifacts and faunal remains attributed to the Middle and Upper Palaeolithic. While the Upper Palaeolithic occupations can be rather precisely dated using radiocarbon, the range of this method (the last 50 millennia, approximately) is too limited to date the Middle Palaeolithic. Unfortunately, OSL dating is rather imprecise and the data obtained for the Middle Paleolithic layers in Gatzarria are very noisy (bottom of Figure~\ref{fig:GatzarriaSummary} -  \citep{DeschampsInPrep}, \citep{Guérin2024datation}): the measured ages scatter between ~35 and ~85 ka and the standard deviations of the probability densities correspond to, on average, $9\%$ of the measured ages. Yet, obtaining a reliable and precise chronology is important for at least two reasons: (i) compare the site with others that yielded similar lithic industries (that is the backbone of the ongoing ERC project Quina World) and (ii) place the occupations of Gatzarria in their environmental context. Indeed, the period between 85 and 35 ka is characterized by major important, worldwide temperature fluctuations (Middle of Figure~\ref{fig:GatzarriaSummary}). In this context, improving the statistical inference that can be drawn from OSL measurements is crucial to discuss co-evolutions of human behavior and peopling dynamics on the one hand, and of climate and environments on the other.
\begin{figure}
    \includegraphics[width = \linewidth]{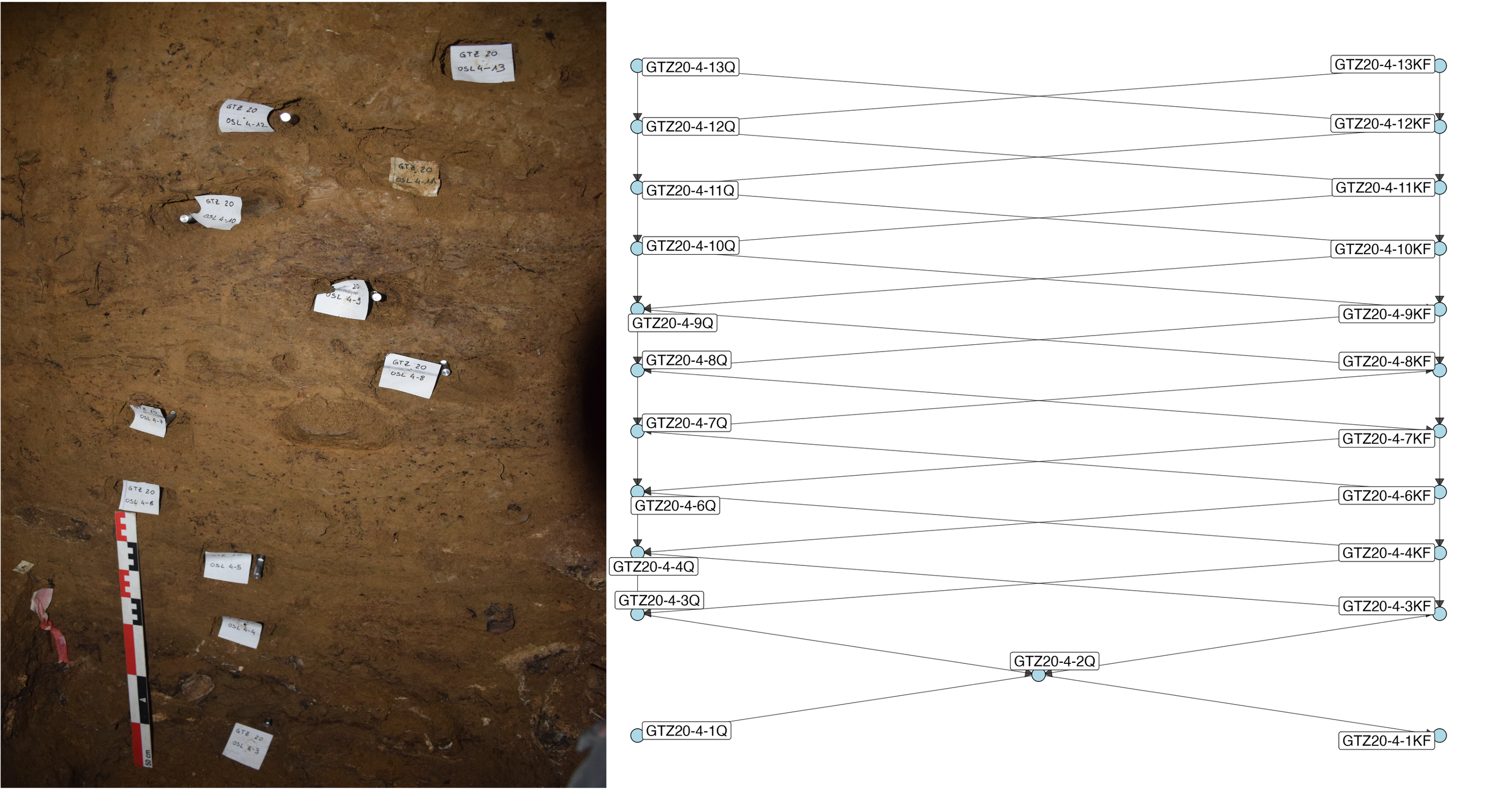}
    \caption{Left: Gatzarria excavation site's picture. Right: Stratigraphic Structure obtained during excavation.}
    \label{fig:DAGGatzarria}
\end{figure}

Application of the isotonic projection is made possible by the ordering constraints shown in Figure~\ref{fig:DAGGatzarria}: all samples were taken from the same stratigraphic column, and each sample (except GTZ20-4) was measured using two methods: OSL of quartz and its equivalent IRSL (for Infra-Red Stimulated Luminescence) of K-feldspar (for measurement procedures and details, the reader is referred to \citep{DeschampsInPrep} and  \citep{Guérin2024datation}).
Figure~\ref{fig:GatzarriaSummary} shows the comparison, for each age, between the measured credible intervals (bottom) and the projected ones through isotonic regression (top). Two observations are most striking: 1. Most of the ages become tightly clustered around 58-65 ka; 2. the length of the credible intervals is greatly reduced. Indeed, the age uncertainties (expressed as the ratio of a quarter of the $95\%$ interval to the age estimate, which is taken as a surrogate for the standard deviation) amount to $3\%$, on average. 
\begin{figure}
    \centering
    \includegraphics[width = .9\linewidth]{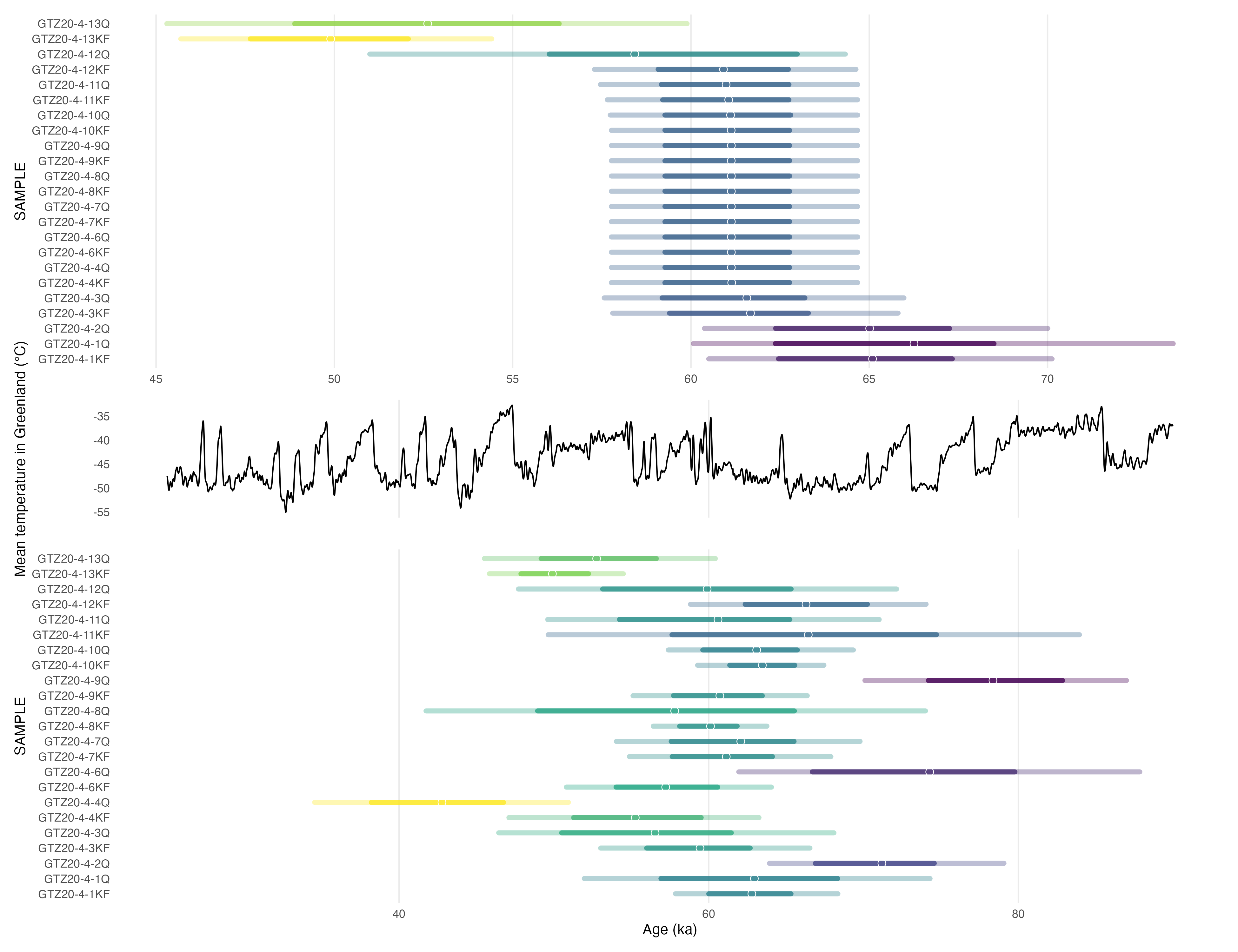}
    \caption{Bayesian modelling for measured luminescence ages at the archaeological site of Gatzarria (SW France). Top: HPD regions at 68\% (opaque) 95\% (transparent) of the Projected model. Middle: Global climatic fluctuations between 90 and 30 ka, as given by the mean annual temperature in Greenland \cite{kindler2014temperature}. Bottom: HPD regions at 68\% (opaque) 95\% (transparent) of the Unconstrained model}
    \label{fig:GatzarriaSummary}
\end{figure}

\noindent Such a level of precision is unprecedented in OSL dating: in a recent review, \cite{murray_optically_2021} acknowledged that the best achievable level precision using OSL is $5-10\%$. Most importantly, the middle part of the sequence (between samples GTZ20-4-4 to 4-9), which has yielded Quina Mousterian \citep{Bourguignon1997Le},\citep{DeschampsInPrep}, is dated to 58-65 ka at the $95\%$ credibility level (59-63 ka at $68\%$). Thus, these time intervals most likely point to a major climatic degradation known as Heinrich Stadial 6 (HS6) that lasted from 64 to 59 ka (\cite{rasmussen2014stratigraphic}), and during which icebergs discharged loads of sediment almost at the latitude of Gatzarria, in the bay of Biscay.

\noindent Much has already been written on the presumed chronology of Quina Mousterian in South West France; some scholars attributed it to Marine Isotope Stage 4 (74-59 ka) based on the abundance of reindeer remains found with Quina-type tools (see, e.g., \cite{mellars1970chronology}, \cite{mellars1991comparison}, \cite{jaubert2009archeosequences}, \cite{discamps2011human}). More recently, based on a regional study comparing micro-mammals and herbivore remains found in Quina contexts, \cite{discamps2017reconstructing} suggested HS6 as the most likely period for the Quina Mousterian; yet, no numerical dating study was able to validate (or invalidate) this hypothesis, simply because the chronological resolution achievable with numerical dating methods was not good enough. Here, for the first time, applying isotonic regression to a set of luminescence ages allows discussing the association between Quina Mousterian and HS6. The availability of this Bayesian model will certainly change chronological interpretations of the Middle Palaeolithic record, by making correlations between human evolution and short yet intense climatic events possible.  

\section{Conclusion}
\label{sec:conclusion}
Classical Bayesian chronological models incorporate the ordering structure directly into the prior. However, such priors are far from neutral. For instance, the Uniform-order prior induces a highly informative distribution on the total duration $d_n = \theta_n - \theta_1,$ systematically favoring larger durations and thereby producing the well-known boundary biases. 

\noindent We adapt a Bayesian framework for parameter estimation under partial order constraints based on the projection of an unconstrained posterior onto an isotonic cone. Rather than specifying an informative prior on the ordering structure as is the case in classical chronological Bayesian models, our approach enforces stratigraphic constraints through a Wasserstein $W_2$ projection. This ensures that the resulting posterior remains as close as possible to the unconstrained distribution. 

\noindent This property is particularly appealing in applications such as archaeology or geochronology, where researchers wish to incorporate stratigraphic information without introducing artificial duration inflation or implicit assumptions about the distribution of ages. The projection framework preserves the data-driven variability of the original model while enforcing monotonicity in the least informative manner possible.

\noindent Nevertheless, the quality of the projected posterior ultimately depends on the quality of the unconstrained posterior. Poor estimates can be partially mitigated through the weighted projection. However, large inversions that strongly contradict the stratigraphic ordering cannot be fully corrected. This limitation suggests investigating more robust objective functions, such as using the $\ell_1$ norm. However, a fundamental property of the operator $T_{\mathcal P}$, is the uniqueness of the solution, which is not always satisfied under the $\ell_1$ loss. 

\noindent The proposed framework is expected to be particularly valuable for high-resolution chronological studies, where a large number of samples are collected over relatively short time periods and the duration-inflation effects of constrained Bayesian models are most pronounced.

\clearpage
\appendix
\section{Proof of Proposition \ref{prop:block_solution} }
\label{proof:block_solution}
\begin{proof}
Let $(\theta,\lambda)$ satisfy the KKT conditions associated with the quadratic program \eqref{eq:QP}. The proof proceeds in two steps. We first see that complementary slackness induces a block structure on the solution. We then sum the stationarity equations over each block to derive the projected value.

\noindent For every vertex $i\in V$, define the sets of outgoing and incoming edges by
\begin{align*}
E_i^+ &= \left\{t\in\{1,\ldots,m\}\mid e_t=(i,j),\ j\in V\right\},\\
E_i^- &= \left\{t\in\{1,\ldots,m\}\mid e_t=(j,i),\ j\in V\right\}.
\end{align*}
Let $C \in \R^{m \times n}$ denotes the reduced constraint matrix, constructed using the canonical basis vectors $\lbrace u_1, \cdots, u_n \rbrace $ of $\R^n$:

\begin{equation}
    \label{eq:edge_matrix}
    \begin{bmatrix}
    (u_{e_1(1)} - u_{e_1(2)})^T \\
    \vdots \\
    (u_{e_k(1)} - u_{e_k(2)})^T \\
    \vdots \\
    (u_{e_m(1)} - u_{e_m(2)})^T \\
\end{bmatrix}
= 
\begin{bmatrix}
    c_1^T \\
    \vdots \\
    c_k^T \\
    \vdots \\
    c_m^T \\
\end{bmatrix}, 
\end{equation}
\noindent the complementary slackness condition reads
\[
\lambda^\top C\theta =0,
\]
which is equivalent to
\[
\lambda_t(c_t^\top\theta)=0,
\qquad t=1,\ldots,m.
\]
Since $c_t^\top\theta=\theta_i-\theta_j$ for the edge $e_t=(i,j)$, we obtain
\[
\lambda_t(\theta_i-\theta_j)=0.
\]
Hence, for every edge $(i,j)\in E$, either
\[
\theta_i=\theta_j,
\]
or
\[
\lambda_t=0.
\]
Since every feasible solution satisfies $\theta_i\leq\theta_j$, each constraint is therefore either active ($\theta_i=\theta_j$) or inactive ($\theta_i<\theta_j$).

\noindent Consequently, the active constraints partition the graph into connected blocks
\[
B_1,\ldots,B_K,
\]
on each of which the projected solution is constant. Denoting by $\theta_{B_o}$ the common value on block $B_o$, we have
\[
\theta_i=\theta_{B_o},
\qquad
\forall i\in B_o.
\]

\noindent We now consider the stationarity conditions. For every vertex $i\in B_o$,
\begin{align*}
2w_i(\theta_i-\theta'_i)
+\sum_{t\in E_i^+}\lambda_t
-\sum_{s\in E_i^-}\lambda_s
=0,
\end{align*}
which becomes
\begin{align*}
2w_i(\theta_{B_o}- \theta_i')
&+\sum_{t\in E_i^+(B_o)}\lambda_t
-\sum_{s\in E_i^-(B_o)}\lambda_s \\
&\qquad
+\left[
\sum_{t\in E_i^+\setminus E_i^+(B_o)}\lambda_t
-
\sum_{s\in E_i^-\setminus E_i^-(B_o)}\lambda_s
\right]
=0.
\end{align*}

The last bracket only involves edges connecting $B_o$ to another block. Such constraints are inactive; otherwise the neighboring vertex would also belong to $B_o$. Therefore, complementary slackness implies
\[
\lambda_t=0,
\qquad
\forall t\in
\bigl(E_i^+\setminus E_i^+(B_o)\bigr)
\cup
\bigl(E_i^-\setminus E_i^-(B_o)\bigr).
\]

Summing the stationarity equations over all vertices of $B_o$ yields
\[
\sum_{i\in B_o}
2w_i(\theta_{B_o}- \theta_i')
+
\sum_{i\in B_o}
\left(
\sum_{t\in E_i^+(B_o)}\lambda_t
-
\sum_{s\in E_i^-(B_o)}\lambda_s
\right)
=0.
\]

\noindent The dual terms cancel pairwise since every internal edge contributes once with a positive sign and once with a negative sign. Hence,
\[
\sum_{i\in B_o}
w_i(\theta_{B_o}- \theta_i')
=0,
\]
which immediately gives
\[
\theta_{B_o}
=
\frac{\sum_{i\in B_o}w_i\theta_i'}
{\sum_{i\in B_o}w_i}.
\]
This proves the result.
\end{proof}
\section{Proof of Proposition \ref{prop:projected_sampling}}
 \label{proof:projected_sampling}
\begin{proof}
    Since $ \mathcal P$ is the finite intersection of closed subsets of $\mathbb R^n$, it is a closed subset of the Polish space $\mathbb R^n$. Hence, $\mathcal P$ is itself a Polish space.

       Note that the following conditions are satisfied: \begin{itemize}
        \item The unconstrained space is a separable Banach space $(\R^n, \mid\mid . \mid\mid_W)$
        \item the constrained space $\P$ is nonempty, closed and
          convex.
        \item By strict convexity and coercivity of $\theta \mapsto \|\theta - \theta'\|_W^2$ on the closed convex set $\mathcal{P}$, the projection $T_{\mathcal{P}}(\theta')$ exists and is
unique for every $\theta' \in \R^n$, so that $T_{\mathcal{P}}$ defines a map from $\mathbb{R}^n$ to $\mathcal{P}$. Moreover, as a metric projection onto a closed convex set, $T_{\mathcal{P}}$ is
$1$-Lipschitz for $\|\cdot\|_W$, hence Borel measurable; the push-forward measure \eqref{eq:pushfoward} is therefore well-defined.
    \end{itemize}
Thanks to all those checked properties, we can use Theorem 2 in \cite{astfalck2026posteriorprojectioninferenceconstrained}. 
 \end{proof}

 \section{Proof of the density expression in \ref{dummy_example}}
 \label{proof:dummy_example}
 \begin{proof}
    Let $h$ be a Borel function over $\R^2$, we can set the following decomposition:

\begin{equation*}
    \E \lbrack h(\theta)\rbrack = \E \lbrack h(\theta ) \1_{\theta_1 < \theta_2}\rbrack +
    \E \lbrack h(\theta) \1_{\theta_1 > \theta_2}\rbrack.
\end{equation*}
We set $\phi_M$ the density function of $\mathcal{N}_2(\mu_M, \sigma_M^2 I_2)$. Since we are in the right part of the plane, we keep the age coordinates $\theta$ unchanged.
\begin{equation*}
     \E \lbrack h(\theta) \1_{\theta_1 < \theta_2}\rbrack = \int_{\R^2} h(\theta_1, \theta_2) \1_{\theta_1 < \theta_2} \phi_M(\theta_1, \theta_2) d\theta_1 d\theta_2.
\end{equation*}
In the other case $\theta_1 > \theta_2$, the coordinates are in the other part of the plane, and will be dragged to the diagonal $D$:
\[ \theta = \frac{\theta_1+\theta_2}{2}  \begin{pmatrix}
    1 \\ 1
\end{pmatrix}.\]
Thus, the second expectation would become: 
$$
\E \lbrack h(\theta ) \1_{\theta_1 > \theta_2}\rbrack =
  \E \lbrack h\bigg(  \frac{\theta_1+\theta_2}{2},
  \frac{\theta_1+\theta_2}{2}\bigg)  \1_{\theta_1 > \theta_2}\rbrack
  $$
  We define
 $$ \begin{cases}
        U &= \frac{\theta_1+\theta_2}{2} \\
        V&= \frac{\theta_1-\theta_2}{2} \
    \end{cases} 
    $$
The couple $(U,V)$ is a Gaussian vector with mean 
$\bigg(\frac{ \mu_{M_1} + \mu_{M_2}}{2} ; \frac{ (\mu_{M_1} - \mu_{M_2})}{2}\bigg)$ and variance $\sigma_M^2 /2$, hence the variables $U$ and $V$ are independent.\\
Under the change of variables $(\theta_1,\theta_2)\mapsto(U,V)$, the support
\[
\{\theta_1>\theta_2\}
\]
is mapped onto
\[
\{(U,V)\in\mathbb{R}\times(0,\infty)\},
\]
that is, $U\in\mathbb{R}$ and $V>0$.
Therefore, denoting by $f_U$  and $f_V$ the densities of $U$ and $V$, respectively, we have
    \begin{align*}
  \E \lbrack h(\theta ) \1_{\theta_1 > \theta_2}\rbrack 
    &= \int_{\R} h(u,u) f_U(u) \bigg( \int_{0}^{+\infty} f_V( v) dv  \bigg) du:=  \int_{\R} h(u,u)g_M(u) du,
\end{align*}
where 
$$
g_M(u) = f_U(u) \mathbb{P}(V>0)  = f_U(u) \Phi\bigg(\frac{ \mu_{M_1} - \mu_{M_2}}{\sqrt{2} \sigma_M}\bigg). 
$$

We can see that the density on the region outside $\P$, is the probability being outside $\P$ times a Gaussian ($\mathcal{N}(\frac{\mu_{M_1} + \mu_{M_2}}{2}, \frac{\sigma_M^2}{2})$) density over the diagonal $D$ and where the highest probability is at the centroid of respective means of each dimension $(\mu_{M_1}, \mu_{M_2})$.
\end{proof}

\section{Proof of Theorem  \ref{thm:posterior_simplex}}
\label{proof:posterior_simplex}
  By an affine change of variables, we may assume without loss of
    generality that $[\alpha,\beta]=[0,1]$ 
The proof of the Theorem \ref{thm:posterior_simplex} relies on two
following  elementary lemmas. Lemma  \ref{L1} provides a uniform
stochastic bound for linear forms over the monotone cone. The random
fluctuation term grows only at the order $\sqrt n$, uniformly over all monotone
vectors. Lemma \ref{L2} shows that, under any distribution with a
positive density on $[0,1]$, the order statistics almost surely spread
over the whole interval as the sample size increases.

Hereafter, we use the notation $M_n=O_{\mathbb P}(u_n)$ to mean that
$M_n/u_n$ is bounded in probability.

\begin{lemma}\label{L1}
Let \((\varepsilon_i)_{i\ge1}\) be independent and identically
distributed  centered $L^2$ random variables.  
Then, 
$$
\sup_{0\le \theta_1\le\cdots\le \theta_n\le1} \left| \sum_{i=1}^n\varepsilon_i\theta_i \right| = O_{\mathbb P}(\sqrt n).
$$
\end{lemma}

\begin{proof}
Denote  $S_k=\sum_{i=1}^k\varepsilon_i$ for $k\geq 1$ and  $S_0=0$. By summation by parts, we have 
$$
\sum_{i=1}^n\varepsilon_i\theta_i = S_n\theta_n- \sum_{k=1}^{n-1}S_k(\theta_{k+1}-\theta_k).
$$
For all $( \theta_1, \cdots, \theta_n) \in \cal P $, we get 
\begin{align*} 
\left| \sum_{i=1}^n\varepsilon_i\theta_i
\right|
& \le
|S_n|\theta_n+
\max_{1\le k\le n-1}|S_k|
          \sum_{k=1}^{n-1}(\theta_{k+1}-\theta_k) \\
  & =|S_n|\theta_n+
\max_{1\le k\le n-1}|S_k|
(\theta_{n}-\theta_1)  \le
2\max_{1\le k\le n}|S_k|.
\end{align*} 
By Kolmogorov's maximal inequality (see \cite{Billingsley2012}), 
$$
\max_{1\le k\le n}|S_k| = O_{\mathbb P}(\sqrt n).
$$
Consequently,
$$
\sup_{0\le \theta_1\le\cdots\le \theta_n\le1} \left| \sum_{i=1}^n\varepsilon_i\theta_i \right| = O_{\mathbb P}(\sqrt n).
$$

\end{proof}

\begin{lemma}\label{L2}
Let \(Y_1,\ldots,Y_n\) be i.i.d. random variables with positive density
\(q\) on \([0,1]\).  Denote $Y_{(1)}$ (resp.  $Y_{(n)}$) the minimum
    (resp maximum)  of the sample.    Then, for every
    \(\delta\in(0,1)\), there exists \(c_\delta>0\) such that 
\begin{equation} 
\mathbb P\left(Y_{(n)}-Y_{(1)}\le 1-\delta\right) \le 2e^{-c_\delta n}.
\end{equation}
\end{lemma}

\begin{proof}
Fix $\delta\in(0,1)$. Observe that
$$
\{Y_{(n)}-Y_{(1)}\le 1-\delta\}
\subset
\left\{Y_i\notin[0,\delta/2)\ \forall i\right\}
\cup
\left\{Y_i\notin(1-\delta/2,1]\ \forall i\right\}.
$$
Hence,
\begin{align*}
\mathbb P\!\left(Y_{(n)}-Y_{(1)}\le 1-\delta\right)
&\le
\left(1-\int_0^{\delta/2}q(x)\,dx\right)^n
+
\left(1-\int_{1-\delta/2}^{1}q(x)\,dx\right)^n  \\
&=: q_1^n+q_2^n
\le
2\max(q_1,q_2)^n,
\end{align*}
where $0<q_1,q_2<1$ because the density $q$ is positive on
$[0,1]$. Therefore,
$$
\mathbb P\!\left(Y_{(n)}-Y_{(1)}\le 1-\delta\right)
\le
2e^{-c_\delta n},
$$
for the positive constant $ c_\delta=-\log\!\bigl(\max(q_1,q_2)\bigr).$
This completes the proof.
\end{proof}

We now combine the previous two lemmas to establish the asymptotic
behavior of the constrained posterior distribution when all the true
ages are equal.

\begin{proof}{Theorem \ref{thm:posterior_simplex}}
The posterior density is proportional to
\begin{align*}
\pi_{\cal P} ( \theta_1,...,\theta_n | M) & \propto \exp\left\{ -\frac1{2\sigma^2 }\sum_{i=1}^n  (M_i-\theta_i)^2
\right\} \1_{\mathcal P }(\theta) \\ & \propto \exp\left\{
  -\frac1{2\sigma^2 }\sum_{i=1}^n  (\theta_i-\theta^0)^2 + \frac1{\sigma}  \sum_{i=1}^n\varepsilon_i(\theta_i-\theta^0) \right\} \1_{\mathcal P }(\theta).
\end{align*}

Define
$$
R_n(\theta)=\sum_{i=1}^n \varepsilon_i(\theta_i-\theta^0),
$$
and let $g_n$ be the probability density such that
$$
g_n(\theta)
\propto
\exp\!\left\{
-\frac{1}{2\sigma^2}
\sum_{i=1}^n(\theta_i-\theta^0)^2
\right\}
\mathbf 1_{\mathcal P}(\theta), 
$$
which  is the
joint density of the order statistics of $n$ independent and identically
distributed random variables from the Gaussian distribution
$\mathcal N(\theta^0,\sigma^2)$ truncated to the interval $[0,1]$.
By Lemma~\ref{L2}, for every $\delta\in(0,1)$,
$$
\int_{\{\theta_n-\theta_1\le 1-\delta\}}
g_n(u_1,\ldots,u_n)\,du
\le
2e^{-c_\delta n}.
$$

Moreover,
$$
\sup_{0\le\theta_1\le\cdots\le\theta_n\le1}
\left|
\sum_{i=1}^n\varepsilon_i(\theta_i-\theta^0)
\right|
\le
\sup_{0\le\theta_1\le\cdots\le\theta_n\le1}
\left|
\sum_{i=1}^n\varepsilon_i\theta_i
\right|
+\theta^0\left|\sum_{i=1}^n\varepsilon_i\right|.
$$

Using Lemma \ref{L1} and  $\sum_{i=1}^n \varepsilon_i  =O_{\mathbb
  P}(\sqrt n)$, it follows that 
\begin{equation}\label{mn}
M_n := \sup_{A\in\mathcal P } |R_n(A)| = O_{\mathbb P}(\sqrt n).
\end{equation}

The posterior distribution can therefore be written as
$$
\pi_{\cal P}(\theta \mid M) = \frac{ \exp\{R_n(\theta)/\sigma \} g_n(\theta) }{ \int_{\mathcal P }
  \exp\{R_n(u)/\sigma \} g_n(u) \,du}.
$$

We fix  $\delta \in (0,1 ) $. The posterior probability of the event  $[\theta_n
-\theta_1 \le1-\delta ]$ is 
\begin{align*}
\Pi_{\cal P}( \theta_n -\theta_1 \le1-\delta  \mid M) & = \frac{ \displaystyle \int_{\lbrace \theta_n
    -\theta_1 \le1-\delta \rbrace} \exp\{R_n(\theta)/\sigma \} g_n(\theta) d\theta  }{
  \displaystyle \int_{\mathcal P } \exp\{R_n(\theta)/\sigma \}  g_n(\theta)
  d\theta}  \\ &\leq 
\dfrac {e^{M_n/\sigma} \int_{ \lbrace \theta_n
    -\theta_1 \le1-\delta \rbrace }  g_n(\theta)
                                 d\theta  }{e^{-M_n/\sigma} }\\
                       & \leq 2e^{2M_n/\sigma }e^{-c_\delta n}
\end{align*}
Using \eqref{mn} and  $c_\delta > 0$, we get 
$$
\Pi_{\mathcal P}\!\left(\theta_n-\theta_1\le 1-\delta\mid M\right)
\xrightarrow{\mathbb{P}_0-proba} 0.
$$
\end{proof}

\section{Proof of Theorem \ref{thm:projected_posterior}}
\label{proof:projected_posterior}

\begin{proof}
The unconstrained posterior is the product of Gaussian measures 
$
\Pi(\cdot \mid M) \;=\; \bigotimes_{i=1}^{n} \mathcal{N}(M_i, \sigma^2
).$ In particular, all marginal posterior variances equal $\sigma^{2}$, so the
weights in \eqref{eq:QP} satisfy $w_i=\sigma^{-2}$ for all $i$. The common
factor $\sigma^{-2}$ leaves the minimizer unchanged, and
$T_\mathcal{P}$ reduces to the unweighted isotonic projection. 

Let $(\varepsilon_i)_{i \ge 1}$ be i.i.d.\ standard normal random variables
and $(U_i)_{i \ge 1}$ be i.i.d.\ uniform random variables on $(0,1)$, the
two sequences being independent, all defined on a common probability space
with probability measure $\mathbb{P}_J$. Define, for all $i \ge 1$,
$$
M_i = \theta^0 + \sigma\,\varepsilon_i,
\qquad
\theta'_i = \Phi^{-1}(U_i) \sigma +M_i ,
$$
where $\Phi$  is the cumulative distribution function of the standard
Gaussian distribution  $\mathcal{N}(0 ,   1)$.
By construction, the marginal distribution of $M = (M_1, \dots, M_n)$ under
$\mathbb{P}_J$ is $\mathbb{P}_0$ and, conditionally on $M$, the distribution  
$\theta'_1, \dots, \theta'_n$ is, the unconstrained
posterior $\Pi(\cdot \mid M)$. Setting $\theta=
T_{\mathcal{P}}(\theta')$, the conditional distribution of $\theta$ given
$M$ is the projected posterior $\tilde\Pi(\cdot \mid M)$, whence
\begin{equation}\label{eq:coupling}
\tilde\Pi\bigl(\theta_{\lfloor bn\rfloor} - \theta_{\lfloor an\rfloor}
> \eta \bigm| M\bigr)
= \mathbb{P}_J\bigl(\theta_{\lfloor bn\rfloor} -
\theta_{\lfloor an\rfloor} > \eta \bigm| M\bigr)
\quad \text{almost surely.}
\end{equation}

Since
\[
\theta'_i \;=\; h(\varepsilon_i, U_i)
\;=\; \theta^0 + \sigma\,\varepsilon_i + \sigma\,\Phi^{-1}(U_i),
\]
for a measurable function $h$ not depending on $i$, and the pairs
$(\varepsilon_i, U_i)$ are i.i.d., the variables $\theta'_i$, $i \ge 1$, are
i.i.d.\ under $P_J$ and integrable, with mean
$\mu = \mathbb{E}_J[\theta'_1] = \theta^0$.
Set
$$
Y = \theta' - \mu \mathbf{1}_n,   \qquad
S_k = \sum_{i=1}^{k} Y_i, \qquad S_0 = 0,
$$
$Y_1, \dots$ are i.i.d.with zero mean.
Moreover, we have 
$$
\theta =  T_{\mathcal{P}}(\theta') = \mu \mathbf{1}_n + T_{\mathcal{P}}(Y).
$$

By the max--min characterization of isotonic regression 
\citep{robertson_order_1988}, for $j = \lfloor an \rfloor$,
$$
T_{\mathcal{P}}(Y)_{\lfloor an\rfloor}
= \max_{u \le \lfloor an\rfloor}\, \min_{v \ge \lfloor an\rfloor}\,
\frac{S_v - S_{u-1}}{v - u + 1} .
$$

Taking $u = 1$ in the maximum and $v = n$ in the minimum gives
$$
- \frac{\max_{1 \le k \le n} |S_k|}{\lfloor an \rfloor}
\;\le\;
T_{\mathcal{P}}(Y)_{\lfloor an\rfloor}
\;\le\;
\frac{\max_{0 \le k \le n} |S_n - S_k|}{n - \lfloor an \rfloor + 1} .
$$
By the strong law of large numbers, $S_k / k \to 0$ almost surely, which
implies $\max_{1 \le k \le n} |S_k| / n \to 0$ $\mathbb{P}_J $-almost surely. 
Therefore,  both bounds converge to $0$: 
$$
T_{\mathcal{P}}(Y)_{\lfloor an\rfloor} \longrightarrow 0
\qquad \mathbb{P}_J\text{-almost surely,}
$$
and
$T_{\mathcal{P}}(Y)_{\lfloor bn\rfloor} \to 0$ almost surely. Consequently,
$$
\theta_{\lfloor bn\rfloor} - \theta_{\lfloor an\rfloor}
= T_{\mathcal{P}}(Y)_{\lfloor bn\rfloor}
- T_{\mathcal{P}}(Y)_{\lfloor an\rfloor}
\longrightarrow 0
\qquad \mathbb{P}_J\text{-almost surely,}
$$
and in particular
$$
\mathbb{P}_J\bigl(\theta_{\lfloor bn\rfloor} -
\theta_{\lfloor an\rfloor} > \eta\bigr) \longrightarrow 0 .
$$

Taking expectations in \eqref{eq:coupling} and using the tower property,
$$
\mathbb{E}_0\Bigl[
\tilde\Pi\bigl(\theta_{\lfloor bn\rfloor} - \theta_{\lfloor an\rfloor}
> \eta \bigm| M\bigr)
\Bigr]
= \mathbb{P}_J\bigl(\theta_{\lfloor bn\rfloor} -
\theta_{\lfloor an\rfloor} > \eta\bigr)
\longrightarrow 0 .
$$
Finally, for every $\tau > 0$, Markov's inequality yields
$$
\mathbb{P}_0\Bigl(
\tilde\Pi\bigl(\theta_{\lfloor bn\rfloor} - \theta_{\lfloor an\rfloor}
> \eta \bigm| M\bigr) > \tau
\Bigr)
\le \frac{1}{\tau}\,
\mathbb{E}_0\Bigl[
\tilde\Pi\bigl(\theta_{\lfloor bn\rfloor} - \theta_{\lfloor an\rfloor}
> \eta \bigm| M\bigr)
\Bigr]
\xrightarrow{n\to \infty} 0 ,
$$
that is,
$$
\tilde\Pi\bigl(\theta_{\lfloor bn\rfloor} - \theta_{\lfloor an\rfloor}
> \eta \bigm| M\bigr)
\xrightarrow{\;\mathbb{P}_0-proba\;} 0 ,
$$
which completes the proof.
\end{proof}

\section*{Acknowledgements}
Authors are supported by the Quina World project, ERC starting grant \#851793 awarded to Guillaume Gu\'erin.

\bibliography{DOCTORAT}

\end{document}